\documentclass[sigconf]{acmart}

\usepackage{hyperref}

\usepackage{booktabs} 

\usepackage{multirow}
\usepackage{subcaption}
\usepackage{amsmath}
\definecolor{mygray}{gray}{0.5}
\usepackage{float}
\usepackage{flushend}
\usepackage{mathtools}
\usepackage{soul}
\usepackage{algorithm}
\usepackage{url}

\usepackage{enumitem}

\long\def\comment#1{}
\long\def\comments#1{}

\author{Yifan Qiao, Yingrui Yang, Haixin Lin, Tao Yang}
\affiliation{%
  \institution{Department of Computer Science, University of California}
  \city{Santa Barbara}
  \state{California}
  \postcode{93106}
\country{USA}
}
\email{{yifanqiao, yingruiyang, haixinlin, tyang}@cs.ucsb.edu}

\newcommand{\todo}[1]{\textcolor{red}{TODO: #1}\PackageWarning{TODO:}{#1!}}

\copyrightyear{2023} 
\acmYear{2023} 
\setcopyright{rightsretained} 
\acmConference[WWW '23]{Proceedings of the ACM Web Conference 2023}{May 1--5, 2023}{Austin, TX, USA}
\acmBooktitle{Proceedings of the ACM Web Conference 2023 (WWW '23), May 1--5, 2023, Austin, TX, USA}
\acmDOI{10.1145/3543507.3583497}
\acmISBN{978-1-4503-9416-1/23/04}

\begin{document}

\title{Optimizing Guided Traversal for Fast Learned Sparse Retrieval}

\begin{abstract}

Recent studies show that BM25-driven dynamic index skipping 
can greatly accelerate MaxScore-based document retrieval based on the learned sparse representation derived by DeepImpact. 
This paper investigates the effectiveness of such a traversal guidance strategy during top $k$
retrieval when using other models such as SPLADE and uniCOIL, and  
finds that unconstrained BM25-driven skipping could have a visible relevance degradation 
when the BM25 model is not well aligned with a learned weight model or when retrieval depth $k$ is small. 
This paper  generalizes the previous work and  
optimizes the BM25 guided index traversal with a two-level pruning control scheme and model  alignment 
for fast retrieval using a sparse representation. 
Although there can be a cost of increased latency, the proposed scheme
 is much faster than the original MaxScore method without BM25 guidance  while retaining the relevance effectiveness. 
This paper analyzes the competitiveness  of  this  two-level pruning scheme,
and evaluates its  tradeoff in ranking relevance and 
time efficiency when searching several test datasets.

\end{abstract}

\maketitle
\pagestyle{plain}

\section{Introduction}



\comments{
\begin{figure}[htbp]
\begin{center}
  \includegraphics[width=\columnwidth]{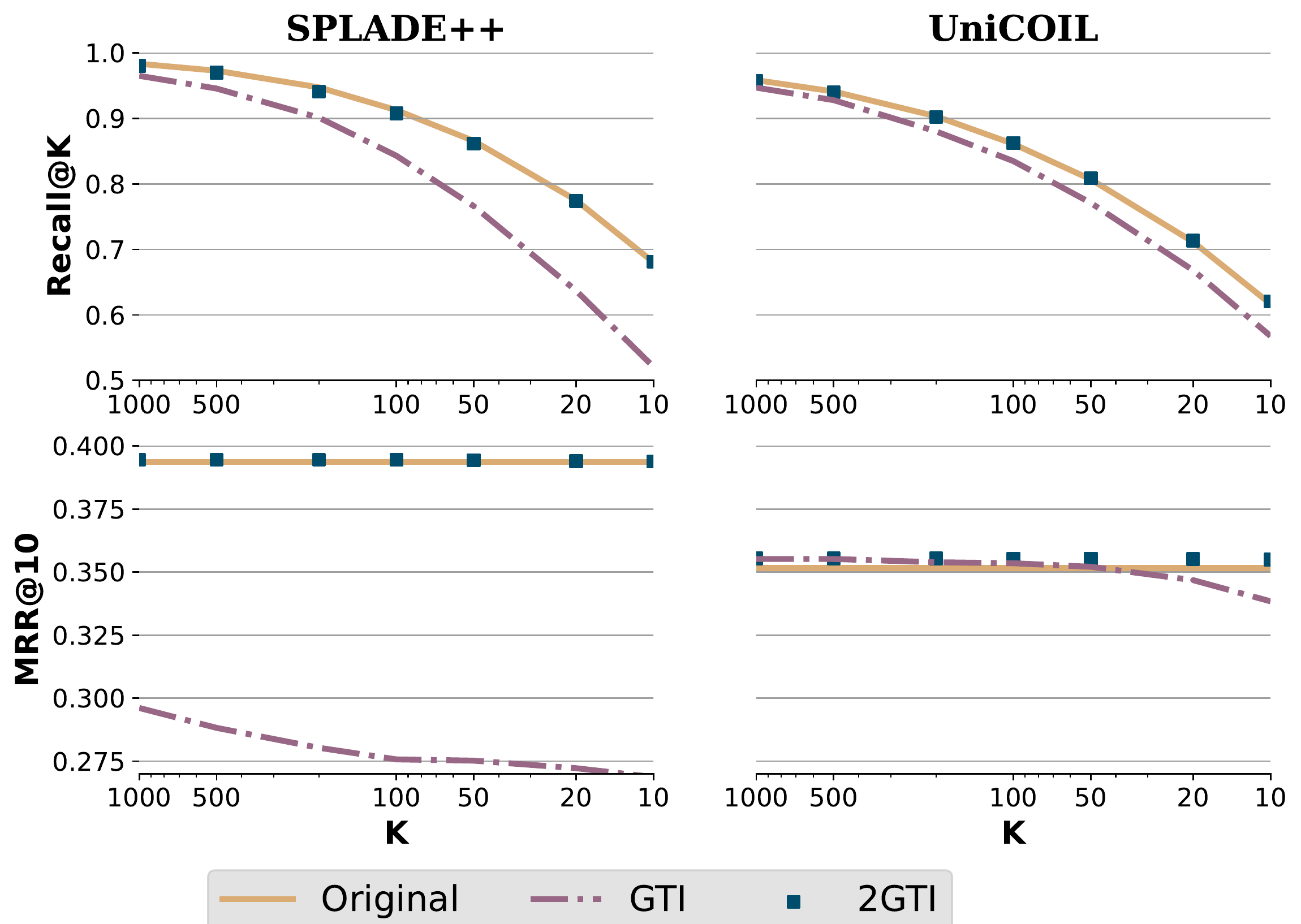}
\end{center}
  \caption{Recall@K and MRR@10 on different retrieval depth K.}
  \label{fig:overview}
\end{figure}
}

Document retrieval for searching a large dataset
often uses a sparse representation of document feature vectors implemented  as 
an inverted index which associating each search term with a list of documents containing such a term. 
Recently learned sparse representations have been developed to
compute term weights using a neural model such as transformer based retriever~\cite{Dai2020deepct, Bai2020SparTerm,
Formal2021SPLADE, Formal2021SPLADEV2,
mallia2021learning, Lin2021unicoil} and
deliver strong relevance results,
together with document expansion (e.g.  ~\cite{Cheriton2019doct5query}). 
A downside is that top $k$ 
document retrieval latency using a learned sparse representation is much large than using the BM25 model  
as discussed in  ~\cite{mallia2021learning, mackenzie2021wacky}. 
In the traditional BM25-based document retrieval with additive ranking, 
a dynamic index pruning strategy based on top $k$ threshold is very effective by computing 
the rank score upper bound on the fly for each visited document 
during index traversal in order to  skip low-scoring documents that are unable to appear in the final top $k$ list.
Well known traversal algorithms  with such dynamic pruning strategies  include    MaxScore~\cite{Turtle1995} and  WAND~\cite{WAND}, and 
their block-based versions Block-Max WAND (BMW)~\cite{BMW} and Block-Max MaxScore (BMM)~\cite{Chakrabarti2011IntervalbasedPF-BMM, Dimopoulos2013OptimizingTD}. 

Mallia et al.~\cite{mallia2022faster}  propose to use BM25 to guide traversal, called GT, 
for fast learned sparse retrieval because the distribution of learned weights results in less pruning opportunities
and they  conducted an evaluation  with retrieval model DeepImpact~\cite{mallia2021learning}. One variation they propose is to compute the final rank scoring as a linear 
combination of the learned weights and BM25 weights, denoted as GTI.
GT is a special case of GTI and  this paper treats GTI as the main baseline.
Since the BM25 weight for a document term pair may not exist in a learned sparse index,
zero filling is used in Mallia et al.~\cite{mallia2022faster} to align the BM25 and learned weight models.
During our evaluation using GT for SPLADE v2 and its revision SPLADE++~\cite{Formal2021SPLADEV2,Formal_etal_SIGIR2022_splade++}, 
we find that 
as retrieval depth $k$ decreases, BM25 driven skipping becomes too aggressive in dropping 
documents desired by top $k$ ranking  based on learned term weights, which 
can cause a significant  relevance degradation.
In addition,  there  is still some room to further improve index alignment of GTI for more accurate
BM25 driven pruning.


To address the above issues, 
we improve our earlier pruning study on dual guidance with combined BM25 and learned weights~\cite{Qiao20222GT}. 
Our work generalizes GTI by constraining the pruning influence of BM25 
and providing an alternative smoothing method to align the BM25 index with learned weights.  
In Section~\ref{sect:maxscore}, we  
propose a two-level parameterized  guidance scheme with index alignment, 
called 2GTI, to manage pruning decisions during  MaxScore based traversal.
We analyze some formal properties of 2GTI on its  relevance behaviors and configuration conditions when 2GTI
outperforms  a two-stage top $k$ search algorithm for a query in relevance. 

Section~\ref{sect:evaluation} and Appendix~\ref{sect:extraeval} 
present an evaluation of 
2GTI with  SPLADE++~\cite{Formal2021SPLADE, Formal2021SPLADEV2,Formal_etal_SIGIR2022_splade++}
and uniCOIL~\cite{Lin2021unicoil, 2021NAACL-Gao-COIL} in addition to DeepImpact~\cite{mallia2021learning} when using MaxScore 
on the MS MARCO datasets.
This evaluation shows  that when retrieval depth $k$ is small, 
or when the BM25 index is not well aligned with the underlying learned sparse representation, 
2GTI can outperform GTI and retain  relevance more effectively. 
In some cases, there is a tradeoff that 2GTI based retrieval may be slower than that of GTI while 2GTI is still much faster than the original MaxScore method without BM25 guidance.
2GTI is also effective for the BEIR datasets in terms of  the zero-shot relevance and retrieval latency. 
In Appendix~\ref{sect:bmw}, we have
extended the use of 2GTI for  a BMW-based algorithm such as VBMW~\cite{Mallia2017VBMW}. 
We demonstrate that 2GTI with VBMW can be useful for a class of short queries and when $k$ is small. 

\section{Background and Related Work}
\label{sect:background}


The top-$k$ document retrieval problem identifies top ranked results in matching a query.
A document representation uses a feature vector to capture  the semantics of a document. 
If these vectors contain much more zeros than non-zero entries, then such a representation is considered sparse.
For a large dataset, document retrieval often uses a
simple additive formula as the first stage of search and it computes the rank score of each document $d$ as:
\begin{equation}
\label{eq:BM25}
\sum_{t \in Q} w_t \cdot  w(t,d),
\end{equation}
where $Q$ is the set  of all search terms,
$w(t,d)$ is a weight contribution of  term $t$  in document $d$, 
and  $w_t$ is a document-independent or query-specific term weight. 
Assume that $w(t,d)$ can be statically or dynamically scaled, 
this paper  views  $w_t=1$ for simplicity of presentation. 
An example of such formula is BM25~\cite{Jones2000} which is  widely used.
For a sparse representation, a retrieval algorithm  often  uses 
an \textit{inverted index} with a set of terms, and a \textit{document posting list} of each term.
A posting record in this  list  contains document ID and  its weight for the corresponding term.

{\bf Threshold-based skipping.} During the traversal of posting lists in document retrieval,
the  previous studies have advocated dynamic pruning strategies 
to skip low-scoring documents, which cannot appear on the final top-$k$ list
~\cite{WAND,strohman2007efficient}.
\comments{
Some information of a posting block $p$
can be accessible without decompression, and such information contains
the maximum weighted term feature score among all documents
and the maximum document ID in this block, denoted  as $BlockMax(p)$ and  $MaxDocID(p)$.
}
To skip the scoring of a document, a pruning  strategy computes the upper bound rank score of a candidate document $d$,
referred to as $Bound(d)$. 

If $Bound(d) \leq  \theta$ where $\theta$ is the rank score threshold in the top final $k$ list, this document can be skipped. 
For example, WAND~\cite{WAND} uses the maximum term weights of documents  of each posting list to determine the rank score upper bound of
a pivot document while BMW~\cite{BMW} and its variants (e.g. ~\cite{Mallia2017VBMW}) optimize WAND 
use  block-based maximum weights to compute the score upper bounds. MaxScore~\cite{Turtle1995} uses 
term partitioning and the top-$k$ threshold to skip unnecessary  index visitation and scoring  computation. 
Previous work has also pursued  a ``rank-unsafe'' skipping strategy by  deliberately over-estimating the current top-$k$ threshold 
by a factor~\cite{WAND, 2012SIGIR-SafeThreshold-Macdonald, 2013WSDM-SafeThreshold-Macdonald, 2017WSDM-DAAT-SAAT}.

{\bf Learned sparse representations.}
Earlier sparse representation studies are conducted  in \cite{Zamani2018SNRM},
DeepCT~\cite{Dai2020deepct}, and SparTerm~\cite{Bai2020SparTerm}. 
\comments{
\citet{Dai2020deepct} learns  contextualized term weights to replace TF-IDF weights.  
}
Recent work on this subject includes 
SPLADE~\cite{Formal2021SPLADE, Formal2021SPLADEV2,Formal_etal_SIGIR2022_splade++}, which learns token importance for  document expansion with sparsity control. 
DeepImpact~\cite{mallia2021learning} learns neural term weights on documents expanded by DocT5Query~\cite{Cheriton2019doct5query}. 
Similarly, uniCOIL~\citep{Lin2021unicoil} extends the work of COIL~\citep{2021NAACL-Gao-COIL} for contextualized term weights. 
Document retrieval with term weights learned from a transformer has been found slow in ~\cite{mallia2022faster, mackenzie2021wacky}. 
Mallia  et al.~\cite{mallia2022faster} state that the MaxScore retrieval algorithm does not efficiently exploit the DeepImpact scores.
Mackenzie et al.~\cite{mackenzie2021wacky} view that the learned sparse term weights are ``wacky'' as they affect document skipping during retrieval thus they
advocate ranking approximation  with score-at-a-time  traversal. 

Our scheme uses  a hybrid combination of  BM25 and  learned term weights, 
motivated by the previous work in composing lexical and neural ranking~\cite{Yang2021WSDM-BECR,
Lin2021tctcolbert, gao2021complementing, ma2021replication,2022LinearInterpolationJimLin}. GTI adopts that for final ranking.  
A key difference in our work  is that hybrid scoring is used for two-level pruning control and  its formula can be different from final ranking.  The multi-level hybrid scoring difference provides an opportunity for additional pruning and its quality control.
Thus the outcome of 2GTI is not a simple linear ranking combination of BM25 and learned weights and 
two-level guided pruning yields a non-linear ensemble effect 
to improve time efficiency while retaining  relevance. 
Our evaluation will include a relevance and  efficiency comparison with MaxScore using a  simple linear combination.

This paper mainly focuses on MaxScore because it  has been shown  more effective for relatively longer queries~\cite{2019ECIRMallia}. 
We also consider VBMW~\cite{Mallia2017VBMW} because
it is generally acknowledged to represent the state of the art~\cite{mackenzie2021wacky} for many cases, 
especially   when $k$ is small  and the query length is short~\cite{2019ECIRMallia}. 

\section{Design Considerations}
\label{sect:design}


\begin{figure}[htbp]
\begin{center}
  \includegraphics[width=\columnwidth]{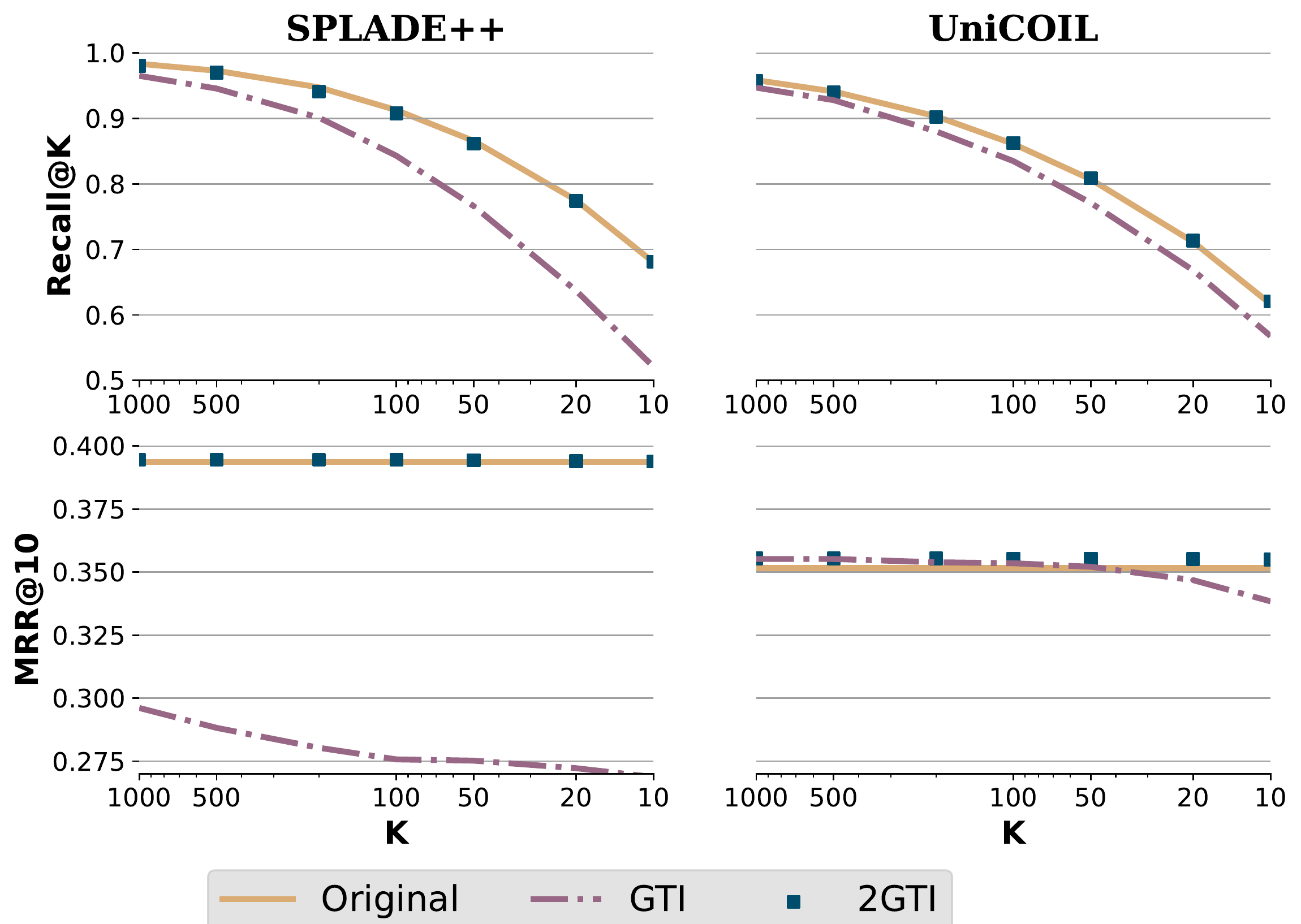}
\end{center}
  \caption{Recall@k and MRR@10 when $k$ varies. }
  \label{fig:overview}
\end{figure}


\comments{
\begin{figure}[htbp]
    \centering
    \includegraphics[width=1\linewidth]{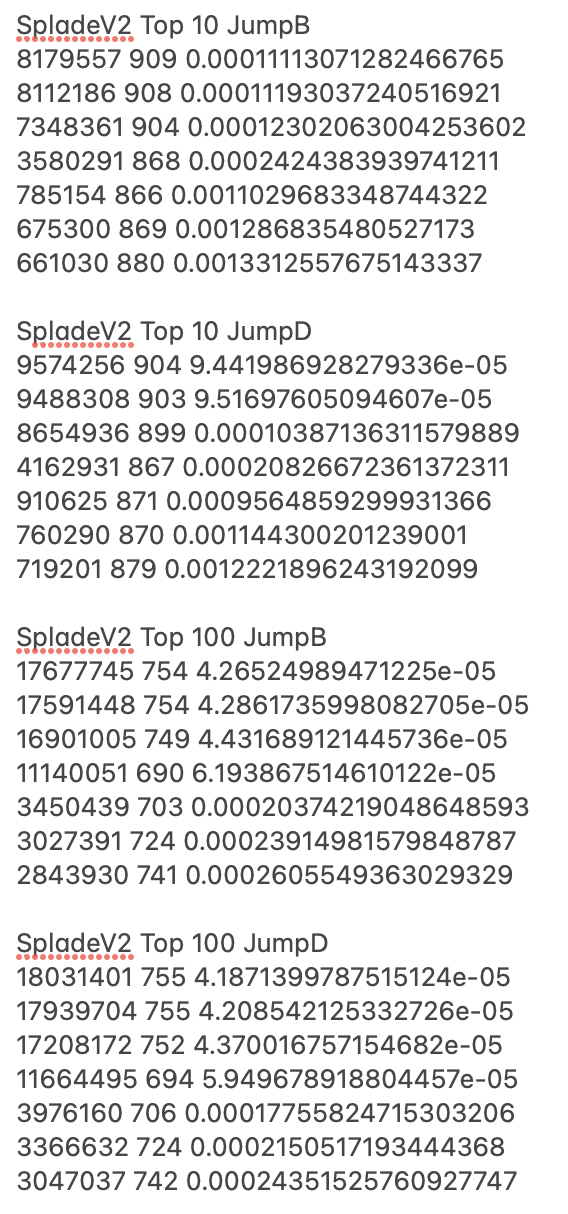}
    \caption{Ratio of making run skipping decisions when guided by BM25.}
    \label{fig:skip}
\end{figure}

Observations:

\begin{enumerate}
    \item With word removal, less doc can be skipped via BM25. The reason is because that the posting length is shorter.
    \item The smaller the alpha is, fewer doc can be skipped. This result indicates that the BM25 scores might be better at skipping.
    \item The final top-k results are almost identical  no matter the choice of alpha, because the accuracy of making wrong decisions is extremely low.
    \item This result cannot guide or explain the selection of the alpha value.
\end{enumerate}
}

Figure ~\ref{fig:overview} shows the performance of the original MaxScore retrieval algorithm without BM25 guidance, GTI,  and the proposed 2GTI scheme in terms of MRR@10 and recall@k when varying top $k$ in searching MS MARCO passages on Dev query set.
Here $k$ is the targeted number of top documents to retrieve and it is also called retrieval depth sometime in the literature.
Section~\ref{sect:evaluation}  has more  detailed dataset and index information.  
For both SPLADE++ and uniCOIL, we build the BM25 model following ~\cite{mallia2022faster} to expand passages  first using DocT5Query,
and then use the BERT's Word Piece tokenizer to tokenize the text, and align the token choices of BM25 with these learned models.
From Figure~\ref{fig:overview}, there are significant recall and MRR drops with GTI
when $k$ varies from 1,000  to 10.  There are two reasons contributing to the relevance drops. 
\begin{enumerate}
[leftmargin=*]
    \item When the number of top documents $k$ is relatively small, the relevance drops significantly. 
As $k$ is small, dynamically-updated top $k$ score threshold becomes closer to the maximum rank score of the best document.
Fewer documents fall into  top $k$ positions  and more documents below the updated  top $k$ score threshold
would be removed earlier. Then the accuracy of skipping becomes more sensitive. 
The discrepancy of BM25 scoring and learned weight scoring can cause good candidates
to be removed inappropriately by BM25 guided pruning, which can lead to a significant relevance drop for small $k$.
  
    
\item  The relevance drop for SPLADE++ with BM25 guided pruning is noticeably much more significant than uniCOIL.
That can be related to the fact that SPLADE++ expands tokens of each document  tokens differently and much more aggressively  than uniCOIL. 
As a result, 98.6\% of term document pairs in SPLADE++ index does not exist in the
BM25 index even after docT5Query document expansion while this number is 1.4\%  for uniCOIL. 
Thus, BM25 guidance can become less accurate and improperly skip  more good documents. 

  
\end{enumerate}
\comments{
This table also shows the alighnment degree between uniCOIL index and BM25-T5, and between SPLADE and BM25-T5.
BM25-T5 index is fairly aligned with uniCOIL index because they use the same docT5Query method, and their tokenization is 
based on BERT's Word Piece tokenizer.
The BM25-T5 index and  SPLADE index are not well aligned because BM25 is builted on docT5Query expansion while SpadeV2 has its own neural
method to learn the token weight, even they are unified under the same Word Piece tokenaization. 
}

With the above consideration, our objective is to control the influence of BM25 weights in a constrained manner for safeguarding relevance prudently,
and to develop better weight alignment when the BM25 index is not well aligned with the learned sparse index.
In Figure~\ref{fig:overview}, the recall@k number of 2GTI marked with blue squares is similar to that of the original method 
without BM25 guidance. Their MRR@10 numbers overlapped with each other,  forming a nearly-flat lines, which indicates
their MRR@10 numbers are similar even $k$ decreases.
The following two sections present our solutions in addressing the above two issues respectively.

\begin{figure*}[!htbp]
\begin{center}
  \includegraphics[width=2\columnwidth]{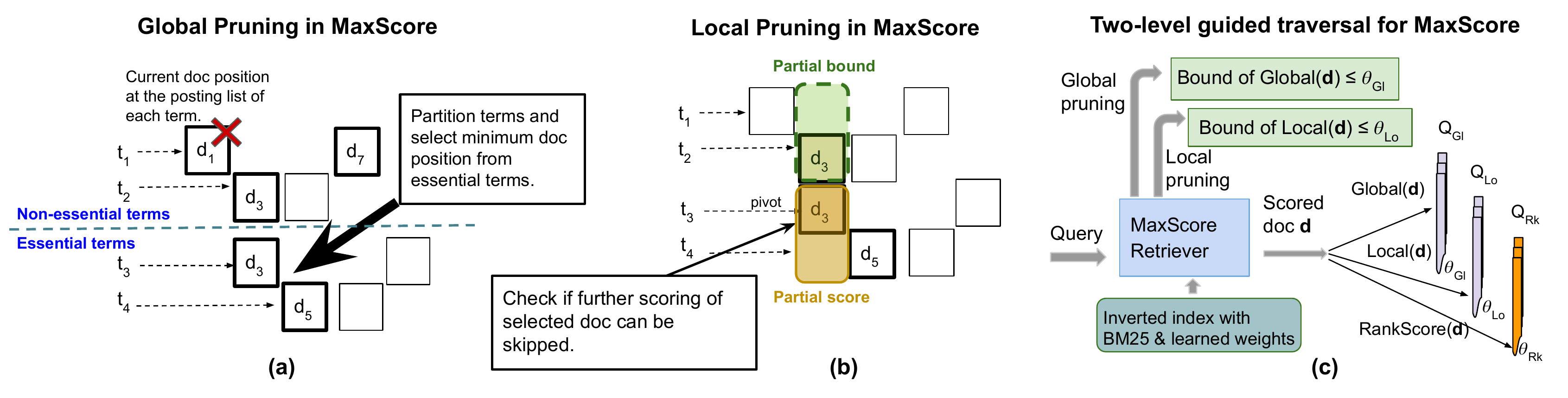}
\end{center}
  \caption{ (a) and (b): Example of two-level pruning in MaxScore. (c) Two-level guided traversal for MaxScore.}  
  \label{fig:maxscoreex}
\end{figure*}

\section{Two-level Guided Traversal}
\label{sect:maxscore}
\subsection{Two-level guidance for MaxScore}
\label{sub:local}

We assume the posting list of  each term is sorted in an increasing order of document IDs in the list.
The MaxScore algorithm~\cite{Turtle1995}   
can be viewed to conduct a sequence of traversal steps and  at each traversal step, it conducts term partitioning and then
examines if scoring of a selected document should be skipped.
We differentiate pruning-oriented actions in two levels as follows.
\comments{
\begin{figure}[htbp]
\begin{center}
  \includegraphics[width=\columnwidth]{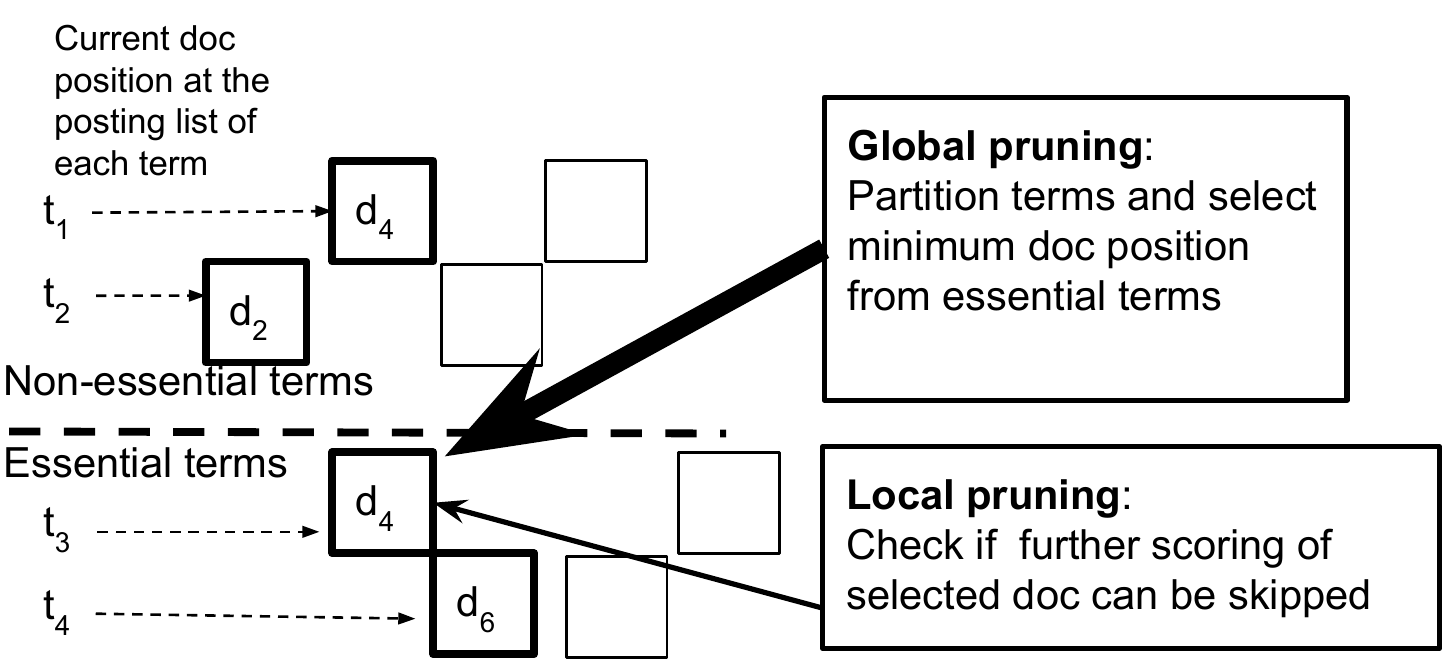}
\end{center}
  \caption{ Example of twol-level pruning in MaxScore}  
  \label{fig:maxscoreex}
\end{figure}
}

\begin{itemize}
[leftmargin=*]
\item {\bf Global level.} 
MaxScore
uses the maximum scores (upper bounds) of each term  and the  current known top $k$ threshold 
to partition terms into two lists at each index traversal step: the essential list and non-essential list.
The documents that do not contain essential terms are impossible to appear in top $k$ results and thus can be eliminated.
In the next step of index traversal, it will start with the minimum unvisited document ID
only from the posting lists of essential terms. 
Thus index visitation is driven by moving such a  minimum document ID pointer from the essential list. 

We consider  this level of pruning as  global  because  
it guides skipping of multiple documents and explores inter-document relationship implied by maximum term weights. 
Figure~\ref{fig:maxscoreex}(a) depicts an example of global pruning flow in MaxScore with 4 terms and each posting list maintains a 
pointer to the current  document being visited at a traversal step. The term partitioning identifies two essential terms  $t_3$ and $t_4$.
The minimum document ID among the current document pointers in these essential terms  
is $d_3$, and any document ID smaller than $d_3$ is skipped from further consideration  during this traversal step. 
The current visitation pointer of the posting list  of non-essential lists also moves to the smallest document ID
equal to or bigger than $d_3$.

\item {\bf Local level.} 
Once a document is selected for possible full evaluation, 
the ranking score upper bound of this document can be estimated and gradually tightened using 
maximum  weight contribution or the actual weight of each query term for this document.
This incrementally refined   score upper bound is compared against the dynamically updated top $k$ threshold, 
which provides another opportunity to fully or partially skip the evaluation of this document. 
We differentiate this level of skipping decision as local because this pruning is localized  towards a specific document selected. 
Figure~\ref{fig:maxscoreex}(b) illustrates an example of local pruning in MaxScore.
$d_3$ is the document selected after term partitioning and the maximum or actual weights contributed from all posting lists for document $d_3$
are utilized for the local pruning decision.

\end{itemize}

Instead of directly using BM25 to guide pruning at  the global and local levels,
we propose to use a linear combination of BM25  weights and learned weights to guide skipping at each level as follows, 
which allows a parameterizable control of their influence.

\begin{itemize}
[leftmargin=*]
\item  We incrementally maintain three accumulated scores for each document $Global(d)$, $Local(d)$, and  $RankScore(d)$. 
$Global(d)$ is for global pruning, $Local(d)$ is for local pruning, and $RankScore(d)$ is for final ranking.
\begin{align}
Global(d) &= \alpha RankScore_B(d)  +(1-\alpha) RankScore_L(d) \nonumber \\
Local(d)  &= \beta RankScore_B(d)  +(1-\beta) RankScore_L(d) \nonumber \\
RankScore(d)  &= \gamma RankScore_B(d)  +(1-\gamma) RankScore_L(d) \nonumber
\end{align}
where $0\leq \alpha, \beta, \gamma \leq 1$, 
$RankScore_B(d)$ follows Expression~\ref{eq:BM25} using BM25 weights, and 
$RankScore_L(d)$ follows Expression~\ref{eq:BM25} using learned weights. 
The RankScore formula follows the  GTI setting in \cite{mallia2022faster}, and 2GTI with $\alpha=\beta=1$  behaves like  GTI.  
2GTI with $\alpha=\beta=\gamma$ is the same as MaxScore retrieval and 
it uses learned neural weights only when $\gamma=0$.


\item With the above three scores for each evaluated document, we maintain three separate queues: 
$Q_{Gl}$,
$Q_{Lo}$,
$Q_{Rk}$ for documents with the $k$ largest scores in terms  of $Global(d)$, $Local(d)$, and $RankScore(d)$ respectively.
The lowest-scoring document in each queue is removed separately without inter-queue coordination. 
These queues  are maintained for different purposes: 
the first two queues regulate  global and local pruning  while the last queue is to produce the final top $k$ results. 
When a document based on local pruning is eliminated for further consideration, this document is not added to global and local queues $Q_{Gl}$ and $Q_{Lo}$. 
But this document may have some partial score accumulated for its $RankScore(d)$,  and it is still added to $Q_{Rk}$ in case this document with the partial
score may qualify in the top $k$ results based on the latest $RankScore(d)$ value.

These three queues yield three dynamic top-$k$ thresholds $\theta_{Gl}$, $\theta_{Lo}$, and $\theta_{Rk}$.
They can be used for a pruning decision to avoid any further  scoring effort  to obtain or refine $RankScore(d)$.

\end{itemize}

{\bf Revised MaxScore pruning control flow:}
Figure~\ref{fig:maxscoreex}(c)  illustrates the extra control flow added for the revised MaxScore algorithm.
Let $N$ be the number of query terms. We define:  
\begin{itemize}
[leftmargin=*]
\item 
Given $N$ posting lists corresponding to $N$ query terms,
each $i$-th posting list contains a sequence of posting records and each record contains document ID $d$, its BM25 weight $w_B(i,d)$ and learned weight $w_L(i,d)$.
Posting records are sorted in an increasing order of their document IDs.
\item An array $\sigma_L$ of $N$ where $\sigma_L[i]$ is the maximum contribution of the learned weight to any document for $i$-th term.
\item An array $\sigma_B$ of $N$ where $\sigma_B[i]$ is the maximum contribution of  the BM25 weight to any document for $i$-th term.
\item $N$ search terms are presorted so that
$\alpha \sigma_B[i] + (1-\alpha) \sigma_L[i] \leq \alpha \sigma_B[i+1] + (1-\alpha) \sigma_L[i+1] $  where $1\leq  i \leq N-1$.
\end{itemize}

{\bf Global pruning with term partitioning.}
For each query term  $1 \leq i\leq N$, we find the largest integer $pivot$ from 1 to $N$ so that 
$\sum_{j=1}^{pivot-1}  (\alpha  \sigma_B[j]+   (1-\alpha) \sigma_L[j])  \leq \theta_{Gl}$.
 All terms from $pivot$  to $N$ are considered as {\em  essential.} 
If a document $d$ does not contain any essential term, the upper bound of 
$Global(d) \leq 
\sum_{j=1}^{pivot-1}  \alpha  \sigma_B[j]+   (1-\alpha) \sigma_L[j]  \leq \theta_{Gl}$.
This  document cannot appear in the final top $k$ list based on the global score.
Then this document is skipped without appearing in any of the three queues.

Once the essential term list above the $pivot$ position is determined, let the next minimum document ID among the current position
pointers in  the posting lists of all essential terms  be  document $d$. We also call it the $pivot$ document. 

{\bf Local pruning.} Next  we check if  the detailed scoring of the selected pivot document $d$ can be avoided fully or partially. 
Following an implementation in \cite{tonellotto2018efficient}, we describe this procedure with a modification to use hybrid scoring  as follows
and it repeats  the following three steps with the initial value of  term position $x$ as the $pivot$ position and $x$ decreases by 1 at each loop iteration.
\begin{itemize}
[leftmargin=*]
\item Let  $PartialScore_{Local}(d)$ be the sum of all term weights of document $d$ in the posting lists from position $x$ to $N$ after linear combination. 
Namely $PartialScore_{Local}(d)= \sum_{i=x}^{N-1}   \beta  w_B(i,d) + (1-\beta)  w_L(i,d)$  when $i$-th posting list contains $d$,  and otherwise this value is 0. 

As $x$ decreases, the term weight of pivot document  $d$  is extracted from the posting list of $x$-th term if available. 
\item Let  $PartialBound_{Local}(d)$  be the bound for partial local score of document $d$ in the posting lists of the first to $x$-th query terms. 
\[
 PartialBound_{Local}(d) = \sum_{j=1}^{x}   \beta  \sigma_B[j]+   (1-\beta) \sigma_L[j].  
\]
\item At any time during the above calculation, if 
\[
PartialBound_{Local}(d) + PartialScore_{Local}(d)  \leq \theta_{Lo}, 
\]
further rank scoring for $pivot$ document $d$ is skipped and
this document will not  appear  in any of the three queues.
Figure~\ref{fig:maxscoreex}(b) depicts that the partial bound  and  partial score of $Local(d_3)$ for pivot document $d_3$
are computed to assist a pruning  decision.

\end{itemize}

\comments{
{\bf Top $k$ thresholding in $Q_{Rk}$ }. Finally when an unpruned document $d$ is added to Queue $Q_{Rk}$
this document may still be skipped when the real value or   the estimated upper bound of $RankScore(d)$  does 
not exceed  $\theta_{Rk}$. When this document is added $Q_{Rk}$, a non-top $k$ document may be kicked out and threshold
$\theta_{Rk}$ is updated. 

We may also  honor a second pruning condition that the upper bound of $RankScore(d)$  does not exceed  $\theta_{Rk}$ in this local level. Our evaluation finds 
that the benefit of deploying this second condition  is fairly small for the tested datasets and thus we treat this as an option. 
On the other hand, setting $\beta = \gamma$ does not yield the best performance as shown in our evaluation.  
}

{\bf Complexity.} 2GTI’s complexity is the same as MaxScore and GTI. 
The in-memory space cost includes  the space
to host the inverted index involved for this query and the three queues.  
The time complexity is proportional to the total number of posting records involved for a query
multiplied  by $\log k$ for queue updating. 
\comments{
$O( max( |Q| \times L, |D|)   \log k)$  
where $|Q|$ is the number of query tokens involved, $L$ is the average length of 
the posting list for each token, and $|D|$ is the number of documents in the index. 
}

\comments{
If Block-Max MaxScore is used, each The posting list is divided and compressed in a blockwise manner, and 
the  partial upperbound of $Local(d)$ can be further tightened using the block-wise maximum weight if a block of a  posting list contains this document $d$.
Since a previous study~\cite{2019ECIRMallia} indicates Block-Max MaxScore retrieval speed is actually slower than MaxScore under several compression schemes,
we leave this as a possible variation.
}
A posting list may be divided and compressed in a block-wise manner and  Block MaxScore can use 2GT similarly while 
a previous study~\cite{2019ECIRMallia} shows Block-Max MaxScore is actually slower than MaxScore under several compression schemes.
We will discuss the use of 2GT in block-based BMW in Appendix \ref{sect:bmw}.




\newtheorem{prop}{Proposition}
\subsection{Relevance properties of 2GTI}
\label{sect:property}

2GTI ensembles BM25 and learned weights for pruning in addition to rank score composition, 
producing a top $k$ ranked list which can be different than additive ranking with 
learned weights or their linear combination  of BM25 weights. 
Thus 2GTI  is not rank-safe compared to any of such baselines. Two-level pruning is driven by
different combination coefficients $\alpha$, $\beta$, and $\gamma$ configured in 2GTI 
and their value gap provides  an opportunity for additional  pruning while 2GTI tries to retain relevance effectiveness. 
Is there  a relevance guarantee   2GTI can offer  in case such pruning skips relevant documents erroneously sometimes?
To address this question analytically, this subsection 
presents three properties regarding the relevance outcome and competitiveness of  the 2GTI based retrieval.
 
Our analysis  will use the following terms. 
Given query $Q$, let $R_x$ be a ranked list of all documents of the given dataset
sorted  in a descend order of their rank scores
based on a linear combination of  their BM25 weights and learned weights with coefficient $x$, 
namely $\sum_{t \in Q} x* w_B(t, d) + (1-x) w_L(t,d)$ for document $d$.
Specifically, there are  three ranked lists: $R_\alpha$, $R_\beta$, and  $R_\gamma$. 
2GTI maintains 3 queues $Q_{Gl}$, $Q_{Lo}$, and $Q_{Rk}$ 
with 3 dynamically updated top $k$ thresholds, $\theta_{Gl}$, $\theta_{Lo}$, $\theta_{Rk}$.
Let $\Theta_{Gl}$, $\Theta_{Lo}$, $\Theta_{Rk}$
 be the final top $k$ threshold of these 3 queues at the end of  2GTI.
Namely it  is the rank score of $k$-th document in the corresponding queue.
The following fact is true:
\[
\theta_{Gl} \leq \Theta_{Gl}, \  
\theta_{Lo} \leq \Theta_{Lo},  \ 
\mbox{ and }
\theta_{Rk} \leq \Theta_{Rk}.
\] 

\begin{prop}
Assume the subset of top $k$ documents in each of $R_{\alpha}$,$R_\beta$, and $R_\gamma$ is unique
after arbitrarily swapping rank positions of documents with the same score.
Then any document that appears in top-$k$ positions of  $R_\alpha$, $R_\beta$, and $R_\gamma$ is in the top-$k$ outcome of 
2GTI.
\end{prop}

\begin{proof}
For any document $d$ that appears in the top $k$ positions of all three ranked lists, 
$
Global(d) \ge \Theta_{Gl} \ge \theta_{Gl}$, $Local(d) \ge \Theta_{Lo} \ge \theta_{Lo}$  and $ RankScore(d) \ge \Theta_{Rk} \ge \theta_{Rk}.$

\comments{
The thresholds of the 3 queues maintained by 2GTI get updated in every index traversal step.
Before  the top $k$-th document in each of  $R_\alpha$, $R_\beta$, and $R_\gamma$ appears in the corresponding queue, 
every iteration of 2GTI satisifies 
\begin{equation*}
\begin{aligned}
\Theta_{Gl} > \theta_{Gl}, 
\Theta_{Lo} > \theta_{Lo}, 
\mbox{ and }  
\Theta_{Rk} > \theta_{Rk}. 
\end{aligned}
\end{equation*}
This is because   based on the assumption of this proposition, 
the rank score of the  top $k$-th document in
each of $R_\alpha$, $R_\beta$, and $R_\gamma$ exceeds that  of top $k+1$-th document in each of them.
Then
\[
Global(d) > \theta_{Gl}, \  Local(d) >  \theta_{Lo}, \mbox{ and }  RankScore(d) > \theta_{Rk}.
\]
}
If document $d$ is eliminated by global pruning during 2GTI retrieval, $Global(d) = \Theta_{Gl} = \theta_{Gl}$
and   
the  $R_{\alpha}$-based rank score 
of  both document $d$ and $(k+1)$-th document in ranked list $R_{\alpha}$ has to be  $\Theta_{Gl}$.  
Then the subset of top $k$ documents in  $R_{\alpha}$ is not unique after arbitrarily 
swapping rank positions of documents with the same score, which is a contradiction.

With the same reason, we can argue  that document $d$ cannot be  eliminated by 
local pruning  or rejected by $\theta_{Rk}$ when being added to $Q_{Rk}$
during 2GTI retrieval. 
Then this document has to appear in the final outcome of 2GTI.
\end{proof}
The following two propositions analyze when  2GTI performs better in relevance than
a two-stage search algorithm called  $R2_{\alpha,\gamma}$
which fetches top $k$ results from  list $R_\alpha$, and then
re-ranks using the scoring formula of $R_\gamma$.
 
\begin{prop}
\label{prop2}
Assume the subset of top $k$ documents in each of $R_{\alpha}$,$R_\beta$, and $R_\gamma$ is unique
after arbitrarily swapping rank positions of documents with the same score.
If 2GTI is configured with $\alpha=\beta$ or $\beta=\gamma$, 
the average $R_\gamma$-based rank score of the top $k$ documents produced 
by 2GTI is no less than that of two-stage algorithm
$R2_{\alpha,\gamma}$.
\end{prop}

\begin{proof}

We let $R2[k]$ denote the top $k$ document subset  in the outcome of $R2_{\alpha,\gamma}$.
To prove this proposition, we compare the average $R_\gamma$-based rank score
of documents in $R2[k]$ and that in $Q_{Rk}$ at the end of 2GTI.
Notice that 
for any document $d$ satisfying $d\in R2[k]$, 
it is in  the top $k$ results of  ranked list $R_\alpha$ and  this top $k$ subset is deterministic
based on the assumption of this proposition. Then  $d$ cannot be  eliminated by global pruning in 2GTI. 

Given  any document $d$ satisfying $d\in R2[k]$ and  $d\not \in Q_{Rk}$ at the end of 2GTI,
it is either eliminated by local pruning with threshold $\Theta_{Lo}$
or   by top $k$ thresholding of Queue $Q_{Rk}$ with threshold $\theta_{Rk}$.
In the later case, $RankScore(d) \leq \theta_{Rk} \leq \Theta_{Rk}$.
When $d$ is eliminated by local pruning, global pruning has to use a different formula because $d$ is not eliminated by global pruning, 
and then 2GTI has to be configured with $\beta=\gamma$ instead of $\alpha=\beta$.
In that case local pruning is identical to elimination with top $k$ threshold of $Q_{Rk}$.
Then $RankScore(d) \leq \theta_{Rk} \leq \Theta_{Rk}$.

\comments{
If 2GTI is configured as $\alpha=\beta$, $Q_{Gl}=Q_{Lo}$, local pruning uses the same formula as global pruning,
and since  $d$ is not eliminated by global pruning, and then it has to be eliminated by 
top $k$ thresholding of Queue $Q_R$ with threshold $\Theta_{R}$.
If 2GTI is configured as $\beta=\gamma$, $Q_{Gl}=Q_{Lo}$, local pruning uses the same formula as global pruning,
because its score upperbound does not exceed $\Theta_R$ at some retrieval iteration of 2GTI.  
$d$ will appear in $Q_{Gl}$ at the end of 2GTI, not eliminated by global pruning. 
It does not stay in Queue $Q_{R}$ at the end 
At that time  $RankScore(d) \leq \Theta_{R} \leq \Theta_\gamma$.
}
Since the size of both $R2[k]$  and $Q_{Rk}$ is $k$,
$|R2[k] - R2[k]\cap  Q_{Rk}|=$ $|Q_{Rk} - R2[k]\cap  Q_{Rk}|$. We can derive:
\begin{equation*}
\begin{aligned}
&\sum_{d \in R2[k] } RankScore(d) \\
=& \sum_{d \in R2[k]  \cap  Q_{Rk}} RankScore(d)
+ \sum_{d \in  R2[k], d \not\in Q_{Rk}} RankScore(d)\\
\leq  & 
 \sum_{d \in R2[k]  \cap  Q_{Rk}} RankScore(d)
+ \sum_{d \in R2[k], d \not\in Q_{Rk}} \Theta_{Rk}\\
=&  \sum_{d \in R2[k]  \cap  Q_{Rk}} RankScore(d)
+ \sum_{d \not\in R2[k], d \in Q_{Rk}} \Theta_{Rk}\\
\leq &  \sum_{d \in R2[k]  \cap  Q_{Rk}} RankScore(d)
+ \sum_{d \not\in R2[k], d \in Q_{Rk}}  RankScore(d)\\
=& \sum_{d \in Q_{Rk}} RankScore(d).\\
\end{aligned}
\end{equation*}

Thus
\[\frac{1}{k}
\sum_{d \in R2[k] } RankScore(d)
\leq  \frac{1}{k}
 \sum_{d \in Q_{Rk}} RankScore(d).
\]

\end{proof}

{\bf Definition 1.} 
For a dataset in which documents are only labeled relevant or irrelevant 
for any test query, we call  ranked list $R_x$ {\em outmatches} $R_y$ 
if whenever $R_y$ orders a pair of relevant and irrelevant documents correctly for a query, $R_x$ also orders them correctly. 

\begin{prop}
Assume documents in a dataset are only labeled as relevant or irrelevant for a test query.
Given a query, when  $R_\gamma$ outmatches $R_\beta$, which outmatches $R_\alpha$,
2GTI retrieves equal or more relevant documents in top-$k$ positions 
than two-stage algorithm $R2_{\alpha, \gamma}$.
\end{prop}
\comments{
\begin{enumerate}
    \item Any document that appears in all top-$k$ lists of $R_\alpha$, $R_\beta$, and $R_\gamma$ is in the top-$k$ outcome of 2GTI. 
    \item Assume each document is labeled as relevant or irrelevant. Given a query, define that ranking algorithm $B$ outperforms $A$ if $A$ orders any pair of relevant and irrelevant documents correctly, $B$ also orders them correctly. When Algorithm $R_\gamma$ outperforms $R_\beta$, which outperforms $R_\alpha$, 2GTI has equal or more relevant documents in top-$k_1$ positions, for any $k_1\leq k$, than a two-stage algorithm with $R_\gamma$ as a re-ranker after top-$k$ $R_\alpha$-based retrieval.
    \item 2GTI($\beta=\gamma$) has the average $R_\gamma$-based rank score of its top $k$ documents no less than that of a two-stage algorithm with $R_\gamma$ as a re-ranker after top-$k$ $R_\alpha$-based retrieval. 
\end{enumerate}
}

\begin{proof}


When 2GTI completes its retrieval for a query,
we count the  number of relevant  documents in top $k$ positions
of 	
list $R_\alpha$,  
queue $Q_{Lo}$, and queue $Q_{Rk}$ as $c_\alpha, c_\beta$, and $c_\gamma$, 
respectively.  To show $c_{\alpha} \leq c_{\beta}$,
we initialize them as  0 first and run the following loop to compute $c_\alpha$ and $c_{\beta}$ iteratively. 
The loop  index variable  $i$ varies from $k$, $k-1$, until $1$, and at each iteration
we look at document $x$ at Position $i$ of $R_\alpha$, 
and  document $y$ at Position $i$ of $Q_{Lo}$.
Let $L_x$ and $L_y$ be their binary label by which value 1 means  relevant and 0 means irrelevant.
\begin{itemize}
\item If $L_x=L_y$, we add $L_x$ to both $c_\alpha$ and $c_\beta$. Continue this loop.
\item Now $L_x \neq L_y$.
If $L_x=0$, $L_y=1$,  we add $1$ to $c_\beta$, and continue the loop. 
If $L_x=1$, $L_y=0$, there are two cases:
	\begin{itemize}
\item If $x$ is within top $i$ positions of current $Q_{Lo}$, 
we add $1$ to both $c_\alpha$ and $c_\beta$. Swap the positions of documents $x$ and $y$ in $Q_{Lo}$. Continue the loop.
\item If $x$ is not within top $i$ positions  of $Q_{Lo}$, since $x$ is in the top 
$k$ of $R_{\alpha}$, it cannot be globally pruned and  it will be evaluated by 2GTI for a possibility of entering $Q_{Lo}$.
If $x$ is ranked before $y$ in list $R_\alpha$, and 
since $R_\beta$ outmatches $R_\alpha$, $x$ has to be ranked before $y$ in both $R_\beta$ and  $Q_{Lo}$. That is a contradiction.
If $x$ is ranked after $y$ in $R_\alpha$,  we swap the positions of  $x$ and $y$ in $R_{\alpha}$. Continue the loop.
	\end{itemize}

\end{itemize}
The above process repeats and moves to a higher position until $i=1$. 
When $i=1$,
with top-1 document $x$ in 
$R_{\alpha}$   
and top-1 $y$ in $Q_{Lo}$,  the only possible cases are $L_x=L_y$ or $L_x=0$ and $L_y=1$.
Therefore, at the end of the above process, $c_\beta \geq c_\alpha$.

Similarly, we can verify that
$c_\gamma \geq c_\beta$ since $R_\gamma$ outmatches $R_\beta$.
Therefore  $c_\gamma \geq c_\beta \geq c_\alpha$.
 The number of relevant documents up to position $k$ retrieved  for 2GTI is $c_\gamma$
 while the number of relevant documents up to position $k$ retrieved  for $R2_{\alpha, \gamma}$
is $c_\alpha$. Thus this proposition is true.
\end{proof}

The above analysis   indicates  that the top  documents agreed by three rankings $R_{\alpha}$, 
$R_\beta$, and $R_\gamma$ are  always kept on the top by 2GTI,
and a properly  configured 2GTI algorithm could outperform  a two-stage retrieval and re-ranking 
algorithm in relevance, especially  
when ranking $R_{\gamma}$ outmatches $R_{\beta}$ and $R_{\beta}$  outmatches $R_{\alpha}$ for a query. 
Since two-stage search  with neural re-ranking conducted after BM25 retrieval is well adopted  in the literature,
this analysis provides useful insight  into the  ``worst-case'' relevance competitiveness of 2GTI with two-level pruning. 
GTI can be considered as 
 a special case of 2GTI with $\alpha=\beta=1$ when the same index is used, and the above three propositions are true for GTI. 
2GTI provides more flexibility in pruning with quality control than GTI and Section~\ref{sect:evaluation} further
evaluates their relevance difference.

\comments{

We only know #relevant results in R_alpha <= #relevants in R_beta, <=#relevant results in R_gamma

Note for the future discussion
We define
R_y strongly outperforms R_x,
when
R_y outperforms R_x, and also for any relevant document found in top $k$ of R_x, 
it will also  appear in top k of R_y.
Then we can prove that
assume documents in a dataset are only labeled as relevant or irrelevant for a test query,
and  $R_\gamma$ strongly outperforms $R_\beta$, which strongly  outperforms $R_\alpha$.
For any $k_1\leq k$, 
2GTI retrieves equal or more relevant documents within top-$k_1$ positions 
than two-stage algorithm $R2_{\alpha, \gamma}$.

}

\comments{
The following two sub-sections solve these problems respectively. In section \ref{sub:local}, we propose a more fine-grained guidance method maintains the same level of relevance as the learned sparse representations, with little extra time cost compared with GT. In section \ref{sub:interpolation}, we introduce two techniques to modify those zero entries in the posting lists to keep the effectiveness of BM25 guided skipping.

}

\subsection{Alignment of tokens and weights}
\label{sect:align}

The BM25 model is usually built on word-level tokenization on the original or expanded document sets and the popular expansion method uses DocT5Query 
with the same tokenization method. When a learned representation uses a different tokenization method such as BERT's WordPiece based on subwords from BERT vocabulary,
we need to align it with BM25 for a consistent term reference.
For example, when using BM25 to guide the traversal of SPLADE index,
the WordPiece tokenizer is used for a document expanded with DocT5Query before BM25 weighting is applied to each token.
Once tokens are aligned, from the index point of view, the same token has two different posting lists based on BM25 weights and based on SPLADE.
To merge them when postings do not align one-to-one, the missing weight is set to zero as proposed in \cite{mallia2022faster}.
We call this zero-filling  alignment.  As alternatives, we propose two more methods to fill missing weights  with better weight smoothness.

\begin{itemize}
[leftmargin=*]
\item {\bf One-filling alignment}. We assign 1 as term frequency for a missing token in the BM25 model 
while this token appears in the learned token list of a document. The justification is that a zero weight is to be too abrupt when such a term is considered to be useful for a document
based on a learned neural model. Having term frequency one means that this token is present in the document, even with the lowest value.
\item {\bf Scaled alignment}.
    This alternative replaces the missing weights in the BM25  model  based on a scaled learned score by
using  the ratio of mean values of non-zero weights in both models.
For document ID $d$ that contains term $t$, let its BM25 weight  be 
$w_B(t,d)$ and its learned weight be $w_L(t,d)$.
Let $w^*_B(t,d)$ be an adjusted BM25 weight. Set
$P_B$ contains all posting records with nonzero  BM25 weights. Set $P_L$ contains posting records with non-zero learned weights. Then $w_B^*(t, d)$ is defined as:

    \begin{equation*}
    w_B^*(t, d)=\left\{
    \begin{aligned}
    w_B(t, d) & , & w_B(t, d) \neq 0, \\
     \frac{\sum_{(t', d')\in P_B} w_B(t', d')/|P_B|}{\sum_{(t', d')\in P_L} w_L(t', d')/|P_L|}w_L(t, d) & , & w_B(t, d)=0.
    \end{aligned}
    \right.
    \end{equation*}


\end{itemize}

\comments{
It is  possibl that there is no learned weight for some terms while their BM25 weights are not zero.
This is because neural weight follows  the tokenizer used in a pretrained language model such as BERT  and the size of
its volcalbulary size can be  smaller than that of the original word-based volcalbulary. 
When BM25 model is well aligned with a learned model, the number of above cases is very small.
But  BM25 model is not well aligned with a learned model such as  SPLADE, the number of above cases is large.
For SPLADE, we have excluded all posting records with zero learned weights and non-zero BM25 weights in the merged index during our evaluation studies.
The reason is that when including all such posting records, the retreval time can be  more than doubled in the tested MS MARCO datasets  while
relevance is not affected as the WordPiece tokenizer is sufficient for the query test sets used.
}

\section{Evaluations}
\label{sect:evaluation}

\comments{
Our evaluation answers the following research questions.
Q1) How does 2GTI compare with the baseline GTI  in relevance effectiveness and latency reduction?

Q2) How does 2GTI compare with the original retrieval without BM25 guidance?

Q3) What is the impact of  different weight filling  alignments  for 2GTI and  also GTI?

Q4) What is the usefulness of  threshold over-estimation?

Q5) What is the impact of adjusting $\alpha$ and $\beta$ parameters in 2GTI?
We also  evaluate  a  design alternative to control the pruning influence of BM25 using threshold under-estimation. 

Q6) How does VBMW-2GTI compare to VBMW-GTI  and MaxScore-2GTI for short queries when $k$ is small? 

}




\begin{table}[h]
\small

\caption{ Model characteristics with MS MARCO Dev set}
\label{tab:statss}

\begin{tabular}{l|cccc}
\hline
\textbf{Index}  & \textbf{Avg. Q Length} & \textbf{\#Postings} & \textbf{Size} & \textbf{Merged} \\ \hline
\multicolumn{5}{l}{ MS MARCO passages} \\ \hline
BM25-T5      & 4.5 (4.5)      &  644M   &   1.2G & -  \\
DeepImpact   & 4.2 (4.2)      &  644M     & 2.6G & 2.6G\\
BM25-T5-B   & 6.6 (6.6)      & 699M       &  1.2G & - \\
UniCOIL      & 6.6 (686.3)    & 592M         & 1.5G & 1.7G\\
SPLADE++ & 23.3 (867.6)     & 2.62B       &  5.6G & 8.3G\\
\hline
\multicolumn{5}{l}{MS MARCO documents} \\ \hline
BM25-T5-B   & 6.8 (7.0)    &  3.39B       & 5.4G & - \\
UniCOIL     & 6.6 (685.0)  & 3.04B       &  7.0G & 8.3G \\
\hline
\end{tabular}
\end{table}


\comments{
\begin{table}[]
\begin{tabular}{l|lllll}
Index      & \# Terms & Avg. Q Length & \# Postings & Mismatch Ratio & Index Size \\
\hline
BM25-T5    & 3,929,111   & 4.5 (4.5)      &  452M   & -             &            \\
DeepImpact & 3,514,102    & 4.5 (4.5)      &  628M     &               & 2.6G       \\
BM25-T5-B  & 27,676    & 6.6 (6.6)      & 699M       & -             & 1.2G       \\
UniCOIL    & 27,678    & 6.6 (686.3)    & 592M        & 1.4\%         & 1.7G       \\
SPLADE v2   & 27,832    & 23.3 (867.6)     & 2.62B       & 98.6\%          & 8.3G       \\
\hline
BM25-T5-B  &  29,144  & 6.8 (7.0)    &  3.39B           & -             &  5.4G     \\
UniCOIL    & 29,141  & 6.6 (685.0)  & 3.04B       & 0.9\%         &  8.3G   
\end{tabular}
\end{table}
}

\noindent
\textbf{Datasets and settings.} 
Our evaluation uses
the MS MARCO document and passage collections~\cite{Craswell2020OverviewOT, Campos2016MSMARCO},
and 13 publicly available BEIR datasets~\cite{thakur2021beir}. 
The results for the BEIR datasets are described in Appendix~\ref{sect:extraeval}. 
For MS MARCO, the contents in the document collections are segmented during indexing and re-grouped after retrieval using ``max-passage'' strategy following~\cite{Lin_etal_SIGIR2021_Pyserini}.
There are 8.8M passages with an average length of 55 words, and 3.2M documents with an average length of 1131 words before segmentation.
The Dev query set for passage and document ranking has 6980 and 5193  queries respectively  with about one judgment label per query. 
Each of the passage/document  ranking task of TREC Deep Learning (DL) 2019 and 2020 tracks provides
43 and 54 queries respectively with many judgment labels per query.

In producing an inverted index, all words use lower case letters. Following GT, we packed the learned score and the term frequency in the same integer. For DeepImpact, we adopt GT's index\footnote{https://github.com/DI4IR/dual-score} directly. The BM25-T5's index is dumped from the DeepImpact index. Both BM25-T5 and DeepImpact are using natural words tokenization.

SPLADE and uniCOIL 
use the BERT's Word Piece tokenizer. In order to align with them, the BM25-T5-B index reported in the following tables uses the same tokenizer as well.
The impact scores of uniCOIL is obtained from Pyserini~\cite{Lin_etal_SIGIR2021_Pyserini}~\footnote{https://github.com/castorini/pyserini/blob/master/docs/experiments-unicoil.md}.
For SPLADE, in order to achieve the best performance, we retrained the model following the setup in SPLADE++~\cite{Formal_etal_SIGIR2022_splade++}.
We start from the pretrained model coCondenser~\cite{coCondenserMarco}
and  
distill using the sentenceBERT hard negatives~\footnote{https://huggingface.co/datasets/sentence-transformers/msmarco-hard-negatives} from a cross-encoder 
teacher~\cite{msmarcoMiniLMv2}
with MarginMSE loss. For FLOP regularization, we use 0.01 and 0.008 for query and documents respectively.  We construct the inverted indexes, convert them to the PISA format,
and compress them using SIMD-BP128~\cite{2015Lemire} following~\cite{2019ECIRMallia, mallia2022faster}.

\begin{table*}[htbp]
    \centering
        \caption{ A Comparison of  2GTI, GTI and the original method with no BM25 guidance for MaxScore}
\resizebox{\textwidth}{!}{%
\label{tab:maxscore-overall}
\renewcommand{\arraystretch}{1.2}
\begin{tabular}{l||cr|cr|cr||cr|cr|cr}
\hline
 & \multicolumn{6}{c||}{$k=10$} &  \multicolumn{6}{c}{$k = 1000$} \\ \hline
 & \multicolumn{2}{c|}{MS MARCO Dev} & \multicolumn{2}{c|}{DL'19} & \multicolumn{2}{c||}{DL'20} & \multicolumn{2}{c|}{MS MARCO Dev} & \multicolumn{2}{c|}{DL'19} & \multicolumn{2}{c}{DL'20} \\
\multicolumn{1}{c||}{\textbf{Method}} & \multicolumn{1}{c}{\textbf{MRR (Recall)}} & \multicolumn{1}{c|}{\textbf{MRT ($P_{99}$)}} & \multicolumn{1}{c}{\textbf{nDCG (Recall)}} & \multicolumn{1}{c|}{\textbf{MRT}} & \multicolumn{1}{c}{\textbf{nDCG (Recall)}} & \multicolumn{1}{c||}{\textbf{MRT}} & \multicolumn{1}{c}{\textbf{MRR (Recall)}} & \multicolumn{1}{c|}{\textbf{MRT ($P_{99}$)}} & \multicolumn{1}{c}{\textbf{nDCG (Recall)}} & \multicolumn{1}{c|}{\textbf{MRT}} & \multicolumn{1}{c}{\textbf{nDCG (Recall)}} & \multicolumn{1}{c}{\textbf{MRT}} \\ \hline \hline

  \multicolumn{13}{l}{\Large \textbf{SPLADE++}, Passages. $\alpha=1$. For 2GTI-Accurate: $\beta=0$; for 2GTI-Fast: $\beta=0.3$ ($k=10$), $0.9$ ($k=1000$). For GTI and 2GTI: $\gamma=0.05$.} \\ \hline

 BM25-T5-B & 0.2611 (0.5179) & 1.7 (8.7) & 0.5931 (0.1556) & 0.8 & 0.5981 (0.2034) & 1.0 & 0.2611 (0.9361) & 9.2 (28.4) & 0.5931 (0.7608) & 6.4 & 0.5981 (0.7641) & 7.4 \\ 
SPLADE++-Org & 0.3937 (0.6801) & 121 (483) & 0.7304 (0.1776) & 135 & 0.7290 (0.2437) & 138 & 0.3937 (0.9832) & 278 (819) & 0.7304 (0.8286) & 317 & 0.7290 (0.8287) & 307 \\
\hline

 -GT & 0.2720 (0.5246) & 121 (438) & 0.6379 (0.1677) & 130 & 0.6106 (0.2192) & 139 & 0.2973 (0.9648) & 336 (1048) & 0.6636 (0.8030) & 330 & 0.6605 (0.8072) & 344 \\
 -GTI & 0.2687 (0.5209) & 118 (440) & 0.6352 (0.1669) & 131 & 0.6083 (0.2190) & 139 & 0.2961 (0.9648) & 332 (1059) & 0.6595 (0.8025) & 318 & 0.6587 (0.8066) & 333 \\

 -2GTI-Accurate & \textbf{0.3939$^\dag$  (0.6812)} & 31.1 (171) & \textbf{0.7401 (0.1846)} & 31.9 & \textbf{0.7278 (0.2480)} & 36.8 & \textbf{0.3946}$^\dag$  (0.9799) & 109 (478) & \textbf{0.7394} (0.8209) & 123 & 0.7297 (\textbf{0.8339}) & 132 \\
 -2GTI-Fast & 0.3934$^\dag$  (0.6792) & \textbf{22.7 (116)} & 0.7380 (0.1837) & \textbf{23.5} & \textbf{0.7278 (0.2480)} & \textbf{26.2} & 0.3937$^\dag$  (0.9662) & \textbf{43.1 (144)} & \textbf{0.7394} (0.8218) & \textbf{42.1} & \textbf{0.7306} (0.8205) & \textbf{45.0} \\ \hline \hline

 \multicolumn{13}{l}{\Large \textbf{UniCOIL}, Passages. $\alpha=1$. For 2GTI-Accurate: $\beta=0$; for 2GTI-Fast: $\beta=0.3$ ($k=10$), $1$ ($k=1000$). For GTI and 2GTI: $\gamma=0.1$.} \\ \hline
 
 BM25-T5-B & 0.2611 (0.5179) & 1.7 (8.7) & 0.5931 (0.1556) & 0.8 & 0.5981 (0.2034) & 1.0 & 0.2611 (0.9361) & 9.2 (28.4) & 0.5931 (0.7608) & 6.4 & 0.5981 (0.7641) & 7.4 \\ 
    UniCOIL-Org & 0.3516 (0.6168) & 10.5 (102) & 0.7027 (0.1761) & 10.4 & 0.6746 (0.2346) & 14.2 & 0.3516 (0.9582) & 35.3 (197) & 0.7027 (0.7822) & 38.6 & 0.6746 (0.7758) & 42.8 \\
    \hline
 -GT & 0.3347 (0.5639) & \textbf{2.1 (11.2)} & 0.6990 (0.1770) & 1.7 & 0.6769 (0.2444) & \textbf{2.4} & 0.3514 (0.9458) & \textbf{10.6} (33.9) & 0.7028 (0.7857) & \textbf{10.3} & 0.6746 (0.7741) & \textbf{10.8} \\
 -GTI & 0.3384 (0.5678) & \textbf{2.1 (11.2)} & 0.6959 (0.1733) & \textbf{1.6} & 0.6739 (0.2422) & \textbf{2.4} & 0.3552 (0.9468) & \textbf{10.6 (33.2)} & \textbf{0.7130 (0.7917)} & \textbf{10.3} & \textbf{0.6899 (0.7857)} & \textbf{10.8} \\
 -2GTI-Accurate & \textbf{0.3550$^\dag$  (0.6205)} & 3.3 (19.2) & \textbf{0.7135 (0.1769)} & 2.4 & \textbf{0.6891 (0.2451)} & 3.4 & \textbf{0.3554}  (0.9566) & 16.9 (68.4) & \textbf{0.7130} (0.7904) & 15.5 & \textbf{0.6899} (0.7823) & 16.5 \\
 -2GTI-Fast & 0.3548$^\dag$  (0.6193) & 2.6 (14.3) & \textbf{0.7135 (0.1769)} & 1.9 & \textbf{0.6891 (0.2451)} & 2.8 & 0.3552 (0.9468) & \textbf{10.6 (33.2)} & 0.7129 (\textbf{0.7917}) & \textbf{10.3} & \textbf{0.6899 (0.7857)} & \textbf{10.8} \\ \hline \hline

 \multicolumn{13}{l}{\Large \textbf{UniCOIL}, Documents. $\alpha=1$. For 2GTI-Accurate: $\beta=0$; for 2GTI-Fast: $\beta=0.5$ ($k=10$), $1$ ($k=1000$). For GTI and 2GTI: $\gamma=0.1$.} \\ \hline

 BM25-T5-B & 0.2716 (0.4749) & 2.7 (13.9) & 0.4246 (0.0741) & 3.6 & 0.4463 (0.1693) & 3.1 & 0.2950 (0.9197) & 12.4 (50.0) & 0.5594 (0.5690) & 15.0 & 0.5742 (0.7148) & 15.7 \\

UniCOIL-Org & 0.3313  (0.5638) & 26.0 (252) & 0.5477 (0.0880) & 28.1 & 0.4996 (0.1920) & 27.3 & 0.3530 (0.9426) & 71.0 (447) & 0.6415 (0.5864) & 79.2 & 0.6059 (0.7502) & 86.0 \\
\hline 

 -GT & 0.3280 (0.5455) & \textbf{6.8 (36.2)} & 0.5199 (0.0779) & 6.2 & 0.4903 (0.1871) & \textbf{7.1} & 0.3530 (0.9361) & 20.9 (73.8) & 0.6445 (0.5842) & \textbf{21.6} & 0.6059 (0.7444) & 22.7 \\
 -GTI & 0.3334 (0.5531) & 6.9 (36.6) & 0.5223 (0.0781) & \textbf{6.1} & 0.4905 (0.1857) & 7.2 & 0.3639 (0.9368) & \textbf{20.6 (72.9)} & \textbf{0.6581} (0.5920) & 21.7 & \textbf{0.6156} (0.7445) & 22.6 \\

 -2GTI-Accurate & \textbf{0.3423$^\dag$  (0.5710)} & 10.2 (56.4) & \textbf{0.5486 (0.0852)} & 8.9 & \textbf{0.4998 (0.1886)} & 10.5 & \textbf{0.3644 (0.9422)} & 33.2 (149) & \textbf{0.6581 (0.5932)} & 32.8 & \textbf{0.6156 (0.7477)} & 37.4 \\
 -2GTI-Fast & 0.3418$^\dag$  (0.5663) & 7.4 (40.4) & 0.5453 (0.0847) & 6.5 & 0.4997 (0.1845) & 8.2 & 0.3639 (0.9368) & \textbf{20.6} (73.0) & \textbf{0.6581} (0.5920) & 21.7 & \textbf{0.6156} (0.7445) & \textbf{22.5} \\ \hline \hline

 \multicolumn{13}{l}{\Large \textbf{DeepImpact}, Passages. $\alpha=1$. For 2GTI-Accurate: $\beta=0$; for 2GTI-Fast: $\beta=0.5$ ($k=10$), $1$ ($k=1000$). For GTI and 2GTI: $\gamma=0.5$.} \\ \hline
 BM25-T5 & 0.2723 (0.5319) & 0.7 (4.7) & 0.6283 (0.1611) & 0.3 & 0.6321 (0.2218) & 0.5 & 0.2723 (0.9348) & 4.9 (21.3) & 0.6283 (0.7704) & 4.6 & 0.6321 (0.7598) & 5.6 \\ 
 DeepImpact-Org & 0.3276 (0.5844) & 9.4 (73.3) & 0.6964 (0.1698) & 4.9 & 0.6524 (0.2035) & 7.8 & 0.3276 (0.9474) & 23.8 (97.0) & 0.6964 (0.7623) & 25.5 & 0.6524 (0.7534) & 38.9 \\
 \hline
 -GT & 0.3276 (0.5805) & \textbf{0.7 (4.7)} & 0.6997 (0.1698) & \textbf{0.3} & 0.6682 (0.2194) & \textbf{0.5} & 0.3276 (0.9454) & \textbf{5.1 (21.1)} & 0.6964 (0.7745) & 4.9 & 0.6527 (0.7574) & 5.8 \\

 -GTI & 0.3375 (0.5866) & \textbf{0.7 (4.7)} & 0.6953(0.1690) & \textbf{0.3} & 0.6846 (\textbf{0.2372}) & \textbf{0.5} & 0.3413 (0.9455) & 5.2 (21.7) & \textbf{0.7072} (0.7871) & \textbf{4.7} & \textbf{0.6854} (0.7745) & \textbf{5.7} \\
 -2GTI-Accurate & \textbf{0.3405$^\dag$ (0.5987)} & 1.1 (8.1) & \textbf{0.7065 (0.1703)} & 0.7 & 0.6850 (0.2361) & 1.2 & \textbf{0.3414} (0.9469) & 7.4 (33.0) & \textbf{0.7072 (0.7875)} & 7.8 & \textbf{0.6854 (0.7778)} & 9.1 \\
 -2GTI-Fast & 0.3395$^\dag$ (0.5934) & \textbf{0.7} (5.1) & 0.7045 (0.1693) & 0.4 & \textbf{0.6882} (0.2371) & 0.6 & 0.3413 (0.9455) & 5.2 (21.7) & \textbf{0.7072} (0.7871) & \textbf{4.7} & \textbf{0.6854} (0.7745) & \textbf{5.7} \\ \hline

\end{tabular}%
}
\end{table*}

Table~\ref{tab:statss} shows the dataset and index  characteristics of the different weighting models 
on the MS MARCO Dev dataset.  
Following \cite{mackenzie2021wacky}, we assume that a query can be pre-processed with  a "pseudo-document" trick that  
assigns custom weights to  query terms in uniCOIL and SPLADE.
Therefore, there may be token repetition in each query to reflect token weighting. 
Column 1 is the mean query length in tokens without or with counting duplicates.  
Column 3 is the inverted index size while the last column is the size after merging BM25 and learned weights in
the index.


The C++ implementation of  2GTI with the modified  MaxScore and VBMW algorithms 
are embedded in PISA~\cite{mallia2019pisa}, and the code will be released in \url{https://github.com/Qiaoyf96/2GTI}.
Our evaluation using this implementation 
runs as a single thread on a Linux server with Intel i5-8259U 2.3GHz  and 32GB memory.
Weights are chosen by sampling queries from the MS MARCO training dataset.

\textbf{Metrics.} 
For MS MARCO Dev set, we report the relevance in terms of 
mean reciprocal rank (MRR@10 on passages and MRR@100 on documents), following the official leader-board standard.
We also report the recall@k ratio which is the percentage of relevant-labeled results appeared in the final top-$k$ results.
For TREC DL test sets, we report  normalized discounted cumulative 
gain (nDCG@10)~\cite{NDCG}. The above reporting follows the common practice of the previous 
work (e.g. ~\cite{mallia2021learning,2021NAACL-Gao-COIL,gao2021complementing,Formal2021SPLADEV2}). 

Before timing queries, all compressed posting lists and metadata for tested queries are pre-loaded into memory, 
following the same assumption in \cite{khattab2020finding, Mallia2017VBMW}.
Retrieval mean response times (MRT) are reported in milliseconds.
The 99th percentile time ($P_{99}$)  is reported within parentheses in the tables below,
corresponding  to the time occurring in the 99th percentile denoted as tail latency in \cite{Mackenzie2018}.

\textbf{Statistical significance.} 
For the reported numbers  on MS MARCO passage and document Dev sets in the rest of this section, 
we have performed a pairwise t-test on relevance difference between 
2GTI and a GTI baseline, and  between  2GTI and the original learned sparse retrieval without BM25 guidance. 
No statistically significant degradation has been observed at the 95\% confidence level.  
We have also performed a pairwise t-test comparing the reported relevance numbers of 2GTI and  GTI and 
mark `$^\dag$' in the evaluation tables if there is a statistically significant improvement by 2GTI over GTI at the
95\% confidence level.
We do not perform a t-test on DL'19 and DL'20 query  sets as the number of  queries in these sets  is small.

{\bf Overall results with MS MACRO.}
Table~\ref{tab:maxscore-overall} lists a comparison of 2GTI with the baseline 
using  three sparse representations for retrieval on MS MARCO and TREC DL datasets.
2GTI uses scaled filling alignment as default while  GTI uses zero filling as specified in \cite{mallia2022faster}.
The $\gamma$ value is   chosen the same for GTI and 2GTI for each representation, which is the best for most of cases.
The ``accurate'' configuration denotes the one that reaches the highest relevance score. 
The ``fast'' configuration denotes the one that reaches a relevance score within 1\% of the accurate configuration
while being much more faster.


{\bf 2GTI vs. GTI in SPLADE++.} Table~\ref{tab:maxscore-overall} shows
2GTI with  default scaled filling significantly outperforms GTI with default zero filling for 
SPLADE++, where BM25 index is not well aligned.
``SPLADE++-Org'' denotes the original MaxScore retrieval performance using SPLADE++ model trained by ourselves
and its MRR@10 number  is higher than what has been reported in ~\cite{Formal_etal_SIGIR2022_splade++}.
When $k=1,000$, GT is slightly better than GTI,
and with the fast configuration, MRR@10 of 2GTI is 32.4\% higher than that of GT while  2GTI is  7.8x faster than GT for the Dev set.
The significant increase in nDCG@10 and decrease in the MRT are also observed in DL'19 and DL'20.
When $k=10$, there is also a large relevance increase  and time reduction
from GTI or GT to 2GTI for all three test sets.
For example, the relevance is 46.4\% or 44.6\% higher and the mean latency is 5.2x or 5.3x faster for the Dev set.

Compared to the original MaxScore method, 
2GTI has about the same relevance score for both $k=10$ and $k=1,000$
while having much smaller latency. For example,   6.5x reduction (278ms vs. 43.1ms) for the Dev passage set when $k=1,000$
and 5.3x reduction when $k=10$ (121ms vs 22.7ms) with the 2GTI-fast configuration.


{\bf 2GTI vs. GTI in DeepImpact and uniCOIL.} As shown in Table~\ref{tab:maxscore-overall},
GTI (or GT)  performs 
very well for $k=1,000$ in both DeepImpact and uniCOIL in speeding up retrieval while maintaining a relevance similar as the original retrieval.
The two-level differentiation for dynamic index pruning does not improve relevance or shorten retrieval time.
This can be explained as  BM25-T5 index is well aligned with the DeepImpact index and with the uniCOIL index.
Also because of this reason, filling to address index alignment is not needed with no improvement  in these two cases.

When $k$ decreases from 1,000 to 10, as shown in Figure~\ref{fig:overview} discussed in Section~\ref{sect:design},
the recall ratio starts to drop, and relevance effectiveness degrades.
When $k=10$ as shown  in Table~\ref{tab:maxscore-overall}, DeepImpact-2GTI-fast 
can increase MRR@10 from 0.3375 by GTI to 0.3395 
for the Dev set and deliver slightly higher MRR@10 or nDCG@10 scores than GTI in DL'19 and DL'19 sets.
For uniCOIL, 2GTI-fast increases MRR@10 from 0.3384 by GTI to 0.3548 for the Dev set and increases nDCG@10 from 0.6959 to 0.7135 for DL'19.
There is also a modest relevance increase for DL'20 passages with $k=10$ and a similar trend is observed  for the document retrieval task. 
The price paid for 2GTI is its retrieval latency increase while its latency is still much smaller than the original retrieval time.

{\bf Design options
with weight alignment and threshold over-estimation.}
Table~\ref{tab:options}
examines the impact of weight alignment and a design  alternative  based on  threshold over-estimation 
 for MS MARCO passage  Dev set using SPLADE++ when $k=10$. 
In the top portion of this table, 
threshold over-estimation by a factor of $F$ (1.1, 1.3, and 1.5)
is used in the original retrieval algorithm without BM25 guidance, and these factor choices are similar as ones in
~\cite{2012SIGIR-SafeThreshold-Macdonald, 2013WSDM-SafeThreshold-Macdonald, 2017WSDM-DAAT-SAAT}. 
That essentially sets $\alpha=0$, $\beta=0$, and $\gamma=0$ while multiplying $\theta_{Gl}$ and $\theta_{Lo}$ by the above factor in 2GTI. 
The result shows that even threshold over-estimation can reduce the retrieval time, relevance reduction is significant, meaning
that the aggressive threshold used causes incorrect dropping of some desired  documents. 

The second portion of  Table~\ref{tab:options} examines the impact of 
different weight filling methods described in Section~\ref{sect:align} for alignment 
when they are applied to GTI and 2GTI, respectively.
In both cases, scaled filling marked as ``/s'' 
is most  effective while one-filling marked as 		``/1''  outperforms zero-filling marked as ``/0''  also.
The MRT of 2GTI/s becomes  10.5x smaller than 2GTI/0
while there is no negative impact to its MRR@10. 
The MRT of GTI/s is about   13.0x smaller than GTI/0 while there is  a large MRR@10 number increase.

\begin{table}[tpbh]
\begin{small}
    \centering
        \caption{Impact of  design options on MS MARCO passages}
\label{tab:filling}
        
        \setlength\tabcolsep{3pt}
        
\label{tab:options}
    \begin{tabular}{l|cc|cc}
    \hline
     SPLADE++. $k=10$ & \textbf{MRR@10} & \textbf{Recall@10} & \textbf{\enspace MRT\enspace} & \textbf{\enspace $P_{99}$\enspace}  \\ 
     \hline \hline
     \multicolumn{5}{l}{Threshold over-estimation} \\ \hline
     Original & 0.3937 &0.6801 & 121 & 483 \\ 
        - $F = 1.1$  & 0.3690 & 0.5707 & 107 & 457\\
        - $F = 1.3$  & 0.3210 & 0.4393 & 95.0 & 420\\
        - $F = 1.5$ & 0.2825 & 0.3670  & 88.2 & 393\\ \hline  
     \multicolumn{5}{l}{Weight alignment for GTI ($\alpha=1, \beta=1, \gamma=0.05)$} \\ \hline
     GTI/0  & 0.2687 & 0.5209 & 118 & 440 \\
        GTI/1  & 0.3036 & 0.5544 & 26.7 & 114 \\
        GTI/s  & 0.3468 & 0.5774 & 9.1 & 36.1 \\ \hline 
     \multicolumn{5} {l}{ Weight alignment  for 2GTI-Accurate ($\alpha=1, \beta=0, \gamma=0.05$)} \\ \hline
    2GTI/0  & 0.3933 & 0.6799  & 328 & 1262 \\
         2GTI/1  & 0.3933 & 0.6818 & 89.3 & 393 \\
         2GTI/s  & 0.3939 & 0.6812 & 31.1 & 171 \\ 
\comments{
    \hline
    \multicolumn{5}{l}{2GTI on Efficient Splade Index VI) BT-SPLADE-L} \\ \hline
    Efficient & 0.3799  & 0.6626 & 17.4 & 59.4 \\
        - 2GTI/s & 0.3772 & 0.6584 & 8.0 & 27.5 \\ \hline

    \multicolumn{5}{l}{2GTI Propositions Verification} \\ \hline
    BM25Retrieve-SPLADERerank & 0.3461 & 0.5179 & - & - \\
    Linear Comb ($\alpha=\beta = \gamma = 0.05$) & 0.3946 & 0.6805 & 120 & 477 \\
        GTI/s ($\alpha = \beta = 1$, $\gamma = 0.05$) & 0.3468 & 0.5774 & 9.1 & 36.1 \\ 
    2GTI/s ($\alpha=1$, $\beta = \gamma = 0.05$) & 0.3939 & 0.6812 & 29.8 & 165 \\

}
     \hline
    \end{tabular}
\end{small}
\end{table}

{\bf A validation on 2GTI's properties.} 
To corroborate the competitiveness analysis in Section \ref{sect:property}, 
Table~\ref{tab:extra} gives MRR@10 scores and retrieval times in milliseconds 
of the algorithms with different configurations 
   on the Dev set of MS MARCO passages with $k=10$ and SPLADE++ weights.
\comments{ 
that the 2GTI frame work with $\alpha=\beta$ (GTI) or $\beta=\gamma$ has better 
relevance than that of the SPLADE re-ranking on top-$k$ documents retrieved by BM25-T5. 
Intuitively, 2GTI acts like a two stage filter. The first stage is a coarse filter based on $\alpha$-combination, and 
lots of $\alpha$-unfavored documents are still kept; the second stage is a fine filter based on 
$\beta$-combination, which removes most of $\beta$-unfavored documents.
For 2GTI, $\beta$ is more leaning to the better learned representation, 
resulting better relevance compared with GTI, where $\beta$ is more leaning to the BM25.
}
The result shows that the listed configurations of  2GTI  
have a higher MRR@10 number than 
2-stage search $R2_{\alpha,\gamma}$, and
2GTI with $\alpha=\beta=1$ that behaves as GTI.  
MRR@10 of ranking with a simple linear combination of BM25 and learned weights
is only slightly higher than 2GTI, but it is much slower.

\begin{table}[tpbh]
\footnotesize
    \centering

        \caption{A validation on  2GTI's properties. $k=10$}
\label{tab:extra}
        \setlength\tabcolsep{3pt}
\label{tab:extra1}
\resizebox{1.01\columnwidth}{!}{%
    \begin{tabular}{l|cc|cc}
    \hline
     & \textbf{MRR@10} & \textbf{Recall@10} & \textbf{ MRT} & \textbf{ $P_{99}$}  \\ 
     \hline 
\comments{
\hline
     \multicolumn{5}{l}{Threshold over-estimation} \\ \hline
     Original & 0.3937 &0.6801 & 121 & 483 \\ 
        - $F = 1.1$  & 0.3690 & 0.5707 & 107 & 457\\
        - $F = 1.3$  & 0.3210 & 0.4393 & 95.0 & 420\\
        - $F = 1.5$ & 0.2825 & 0.3670  & 88.2 & 393\\ \hline  
     \multicolumn{5}{l}{Weight alignment for GTI} \\ \hline
     GTI/0  & 0.2687 & 0.5209 & 118 & 440 \\
        GTI/1  & 0.3036 & 0.5544 & 26.7 & 114 \\
        GTI/s  & 0.3468 & 0.5774 & 9.1 & 36.1 \\ \hline 
     \multicolumn{5} {l}{ Weight alignment  for 2GTI-Accurate ($\beta=0$)} \\ \hline
    2GTI/0  & 0.3933 & 0.6799  & 328 & 1262 \\
         2GTI/1  & 0.3933 & 0.6818 & 89.3 & 393 \\
         2GTI/s  & 0.3939 & 0.6812 & 31.1 & 171 \\ 
    \hline
}
    $R2_{\alpha,\gamma}$ (BM25 retri. SPLADE++ rerank) & 0.3461 & 0.5179 & - & - \\
     GTI/s ($\alpha=\beta=1$, $\gamma=0.05$) & 0.3468 & 0.5774 & 9.1 & 36.1 \\ 
    2GTI/s ($\alpha=1$, $\beta =\gamma = 0.05$) & 0.3939 & 0.6812 & 29.8 & 165 \\
    2GTI/s-Accurate ($\alpha$=1, $\beta$=0, $\gamma$=0.05) & 0.3939 & 0.6812 & 31.1 & 171 \\
    2GTI/s-Fast ($\alpha=1$, $\beta=0.3, \gamma=0.05$) & 0.3934 & 0.6792 & 22.7 & 116 \\
    Linear  comb. ($\alpha=\beta=\gamma=0.05$) & 0.3946 & 0.6805 & 120 & 477 \\

     \hline
    \end{tabular}
    }
\end{table}

\begin{table}[tpbh]

\small
    \centering
        \caption{Use of 2GTI with a new SPLADE model~\cite{lassance2022efficiency}. $k=10$}
\label{tab:extra2}
        \setlength\tabcolsep{3pt}
        
\label{tab:effi}

    \begin{tabular}{l|cc|cc}
    \hline
    BT-SPLADE-L  & \textbf{MRR@10} & \textbf{Recall@10} & \textbf{MRT} & \textbf{$P_{99}$}  \\ 
     \hline 
     Original MaxScore & 0.3799  & 0.6626 & 17.4 & 59.4 \\
    2GTI/s ($\alpha$=1, $\beta$=0.3, $\gamma$=0.05) 
& 0.3772 & 0.6584 & 8.0 & 27.5 \\ 
     GTI/s ($\alpha=\beta=1$, $\gamma=0.05$) 
& 0.3284 & 0.5520 & 6.6 & 24.9 \\ 
     \hline
    \end{tabular}
\end{table}

{\bf Sensitivity on weight distribution.}
We have distorted the SPLADE++ weight distribution in several ways to examine the sensitivity  of 2GTI
and found that 2GTI is still effective.
For example, we apply a square root function to the  neural weight of every token  in MS MARCO passages,
the relevance score of both original retrieval and  2GTI drops to 0.356 MRR@10 due to weight distortion,
while  2GTI is 5.0x faster  than the original MaxScore when $k=10$.

{\bf  Efficient SPLADE model.} 
Table~\ref{tab:effi} shows the application of 2GTI in a recently published 
efficient SPLADE model~\cite{lassance2022efficiency} which has made several improvements in retrieval speed.
We have used the released checkpoint of this efficient  model called BT-SPLADE-L,
which has  a weaker MRR@10 score, but significantly faster than our trained SPLADE baseline reported in Table~\ref{tab:maxscore-overall}.
When used with this new SPLADE model, 
2GTI/s-Fast version  results in a 2.2x retrieval time
speedup over MaxScore. Its MRR@10 is higher than  GTI/s and  
has less than 1\%  degradation compared to the original MaxScore.

\section{Concluding Remarks} 
The contribution of this  paper is a two-level parameterized  guidance scheme with index alignment to optimize 
retrieval traversal with a learned sparse representation. 
Our formal analysis shows that a properly  configured 2GTI algorithm including GTI
can outperform  a two-stage retrieval and re-ranking algorithm in relevance.

Our evaluation 
shows that the proposed 2GTI scheme can make the BM25 pruning guidance  more accurate
to retain the relevance. 
For MaxScore with SPLADE++ on MS MARCO passages, 2GTI can lift  relevance by up-to 32.4\% and is 7.8x faster than GTI when $k=1,000$,
and by up-to 46.4\% more accurate and 5.2x faster when $k=10$.
In all evaluated cases, 2GTI is much faster than the original retrieval without BM25 guidance.
For example, up-to 6.5x faster than MaxScore on SPLADE++  when $k=10$. 
We have also observed similar performance patterns on BEIR datasets when comparing 2GTI with GTI and
the original MaxScore using SPLADE++ learned weights. 
Compared to other options such as  threshold underestimation to  reduce the influence of BM25 weights,
the two-level control is more accurate in maintaining the strong relevance with a much lower time cost. 
While our study is mainly centered with MaxScore-based retrieval, 
2GTI can be used for VBMW and our evaluation shows that  VBMW-2GTI can be a preferred choice for 
a class of short queries without stop words when  $k$ is small.

{\bf Acknowledgments}. We thank  Wentai Xie, Xiyue Wang,  Tianbo Xiong, and anonymous referees
for their valuable comments and/or help. 
This work is supported in part by NSF IIS-2225942
and has used the computing resource of the XSEDE and ACCESS programs supported by NSF.
Any opinions, findings, conclusions or recommendations expressed in this material
are those of the authors and do not necessarily reflect the views of the NSF.

\appendix

\section{Additional Evaluation Results}
\label{sect:extraeval}


\begin{figure}[htbp]
\begin{center}
  \includegraphics[width=1.0\columnwidth]{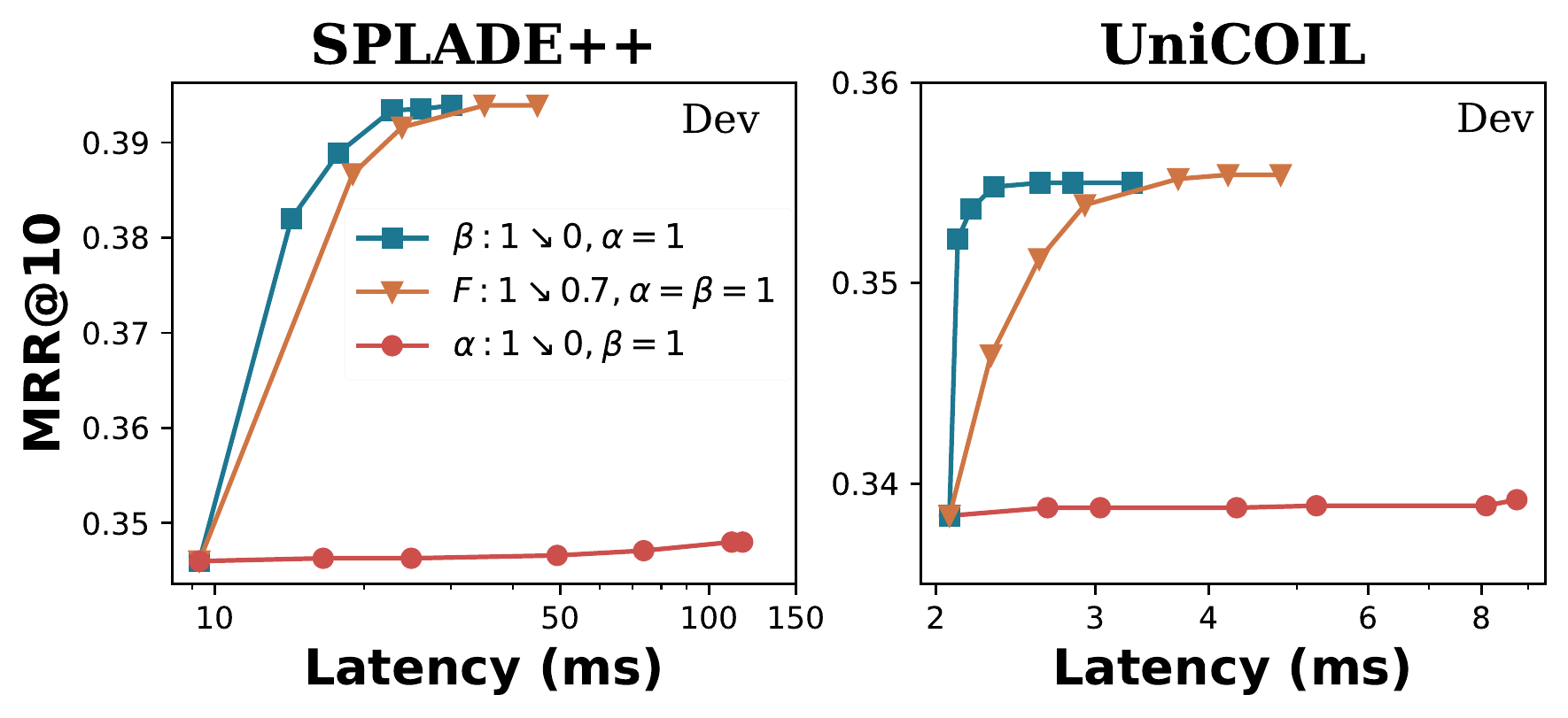}
\end{center}
  \caption{  Controlling influence of BM25 on pruning} 
  \label{fig:local}
\end{figure}

{\bf  Impact of $\alpha$ and $ \beta$ adjustment on 2GTI.}
Figure \ref{fig:local} examines the impact of adjusting parameters  $\alpha$  and
$\beta$ on global and local pruning for the MS MARCO Dev passage test set  when $k=10$
in  controlling the influence of BM25 weights
for SPLADE++ (left) 
and uniCOIL (right).
The $x$ axis corresponds to  the latency increase
while $y$ axis corresponds to the MRR@10 or nDCG@10 increase. 
The results for MS MARCO DL'19, and DL'20 are similar.

The red curve connected with dots fixes $\beta=1$ and varies $\alpha$ from 1 at the left end to  0 at the right end.  
As  $\alpha$ decreases from 1 to 0,  the latency increases because BM25 influences diminish at the global pruning level  and 
fewer documents are skipped. The relevance for  this curve is relatively flat in general
and lower than that of the blue curve, representing the global level BM25 guidance reduces time significantly, while having less impact on the relevance.

The blue curve connected with squares  fixes $\alpha=1$ at the global level and  varies $\beta$ from 1 at the left bottom end to 0 at the right top end.
Decreasing $\beta$ value is positive in general for relevance towards some point as BM25 influence decreases gradually at the local level
and after such a point, the relevance gain becomes much smaller or negative.  For example, after  $\beta$ in the blue curve in SPLADE++ 
becomes 0.3 for the Dev set, its additional decrease does not lift MRR@10 visibly anymore while the latency continues to increase, which 
indicates the relevance benefit has  reached the peak at that point. 
Our experience with the tested datasets is that the parameter  setting for 2GTI can reach a relevance peak  typically when $\alpha$ is close to 1 and $\beta$ varies between 0.3 and 1.

Note that even the above result advocates that 
$\alpha$ is  close to 1,  $\alpha$ and $\beta$ still have different values to be more effective for the tested data, 
reflecting the usefulness
 of two-level pruning control.

{\bf  Threshold under-estimation. }
In Figure \ref{fig:local}, the brown curve connected  with  triangles fixes $\alpha=\beta=1$
and under-estimates the skipping threshold by a factor of $F$ at the local and global levels. 
That behaves like GTI coupled with scaled weight filling as a special case of 2GTI. 
$F$ varies from 1 at the left bottom end to 0.7 at the right top end of this brown curve. 
As $F$ decreases, the skipping threshold becomes very loose and there is less chance  that  desired documents  are skipped. 
Then retrieval relevance can improve while retrieval time can increase substantially. Comparing with the blue curve that adjusts $\beta$,
retrieval takes a much longer time in the brown curve to reach the peak relevance, as shown in this figure, and  
the brown curve is generally placed on the right side of the blue curve. For example on the Dev set with uniCOIL, 
the brown curve with threshold under-estimation reaches  the best relevance  at mean latency  3.7ms 
while the blue curve with $\beta$ adjustment reaches the same peak  at mean  latency 2.3ms, which is 1.6x faster. 
\label{sect:beir}

{\bf  Zero-shot performance on the BEIR datasets.}  
We evaluate the zero-shot ranking effectiveness and  response time  of 2GTI using the 13 search and semantic relatedness  
datasets from the  BEIR 
collection.  
Our training of SPLADE++ model is only based on MS MARCO data without using  any BEIR data. 
Table~\ref{tab:beir} lists the nDCG@10 scores 
of original MaxScore on SPLADE++, 
2GTI/s-Fast ($\alpha$=1, $\beta$=0.3, $\gamma$=0.05)
and GTI ($\alpha$=$\beta$=1, $\gamma$=0.05).  
The retrieval depth is $k=10$ and $k=1000$.
This table also reports mean response time of retrieval in milliseconds. 
The  SPLADE++ model trained by ourself has an  average nDCG@10 score 0.500 
close to 0.507  reported in the SPLADE++ paper~\cite{Formal_etal_SIGIR2022_splade++}.
The original MaxScore's  nDCG@10 score does not change when $k=10$ and $k=1000$. 

\begin{table}[htpb]
\small
    \centering.
    \caption{Zero-shot relevance in NDCG@10 and retrieval latency in milliseconds  on BEIR datasets with SPLADE++}
            \label{tab:beir}
    \begin{tabular}{l|cr|cr|cr}
    \hline
        & \multicolumn{2}{c|}{Original MaxScore} & \multicolumn{2}{c|}{2GTI/s-Fast} & \multicolumn{2}{c}{GTI/s} \\
\textbf{Dataset} & \textbf{nDCG} & \textbf{MRT} & \textbf{nDCG} & \textbf{MRT} & \textbf{nDCG} & \textbf{MRT}\\
 \hline 
        \multicolumn{7}{l}{\textbf{$k$=10} } \\ \hline
        DBPedia & 0.447 & 99.0 & \textbf{0.449} & 34.5 & 0.306 & 10.6 \\
        FiQA & \textbf{0.355} & 5.1 & 0.354 & 3.4 & 0.256 & 0.8 \\
        NQ & \textbf{0.551} & 72.9 & \textbf{0.551} & 28.6 & 0.524 & 7.3 \\
        HotpotQA  & \textbf{0.681} & 453 & \textbf{0.681} & 191 & 0.549 & 46.8 \\
        NFCorpus & \textbf{0.351} & 0.3 & 0.347 & 0.2 & 0.327 & 0.1 \\
        T-COVID & 0.705 & 15.9 & \textbf{0.707} & 9.9 & 0.569 & 2.6 \\
        Touche-2020 & \textbf{0.291} & 8.7 & \textbf{0.291} & 3.2 & 0.237 & 1.3 \\
        ArguAna & 0.446 & 8.8 & 0.448 & 4.0 & \textbf{0.454} & 4.0\\
        C-FEVER & \textbf{0.234} & 635 & 0.231 & 355 & 0.196 & 241\\
        FEVER  & \textbf{0.781} & 1028 & 0.771 & 655 & 0.590 & 160\\
        Quora & \textbf{0.817} & 21.5 & \textbf{0.817} & 6.7 & 0.763 & 1.7\\
        SCIDOCS & \textbf{0.155} & 3.7 & \textbf{0.155} & 2.1 & 0.140 & 1.2\\
        SciFact & \textbf{0.682} & 3.2 & 0.680 & 2.9 & 0.680 & 1.7\\
\hline 
        \textbf{Average } & \textbf{0.500} & - & 0.499 & 2.0x & 0.430 & 6.1x \\
\hline \hline
        \multicolumn{7}{l}{ \textbf{$k$=1000}} \\ \hline
        \textbf{Average} & 0.500 & - & \textbf{0.501} & 2.5x & 0.496 & 2.7x \\
 \hline
    \end{tabular}
\end{table}

When $k=10$, 2GTI has  almost identical nDCG@10 scores as the original MaxScore
while 2GTI is  on average 2.0x faster  than MaxScore for these  BEIR datasets. 
When GTI runs on the same index data, its average nDCG@10 score is 0.43 MRR@10 and 
it is faster than 2GTI 
with an average 6.1x speedup over the original MaxScore for these datasets. 
Two-level pruning in 2GTI can  preserve relevance better than GTI and  
this is consistent with what we have observed for searching MS MARCO passages. 
 
When $k=1000$, the guided traversal algorithms have a better chance to retain relevance. 
2GTI has a slightly higher average relevance of 0.501 MRR@10 than that with $k=10$
and it is about 2.5x faster on average than  the original MaxScore. 
For GTI running on the same index with the same alignment, 
the average  MRR@10 is 0.496 whil average speedup 2.7x over MaxScore. 
Its relevance score is close to that of 2GTI as BM25-driven pruning under a large $k$ value can 
still keep a good recall ratio.

\comments{
{\bf Data sensitivity with distorted weight distributions.}  
We have distorted the SPLADE* weight distribution in several ways to examine the data sensitivity  of 2GTI. 
For example, we apply a square root function to every weight in SPLADE* for MS MARCO passages,
the relevance score of both original retrieval and  2GTI drops to 0.356 MRR@10 due to weight distortion,
but  2GTI is still 5.0x faster  when $k=10$.

}


\section{Two-level guidance for BMW}
\label{sect:bmw}

Two-level guidance can be adopted to control index traversal of a BMW based algorithm such as VBMW as well
because we can also view that such an algorithm  conducts a sequence of index traversal steps,
and can differentiate  its index pruning of each traversal step at the global inter-document and local intra-document levels.  
We use the same symbol notations as in the previous subsection,
assuming the posting lists are sorted by an increasing order of their document IDs. 
We still keep a position pointer in each posting list of search terms to track the current document ID $d_{t_i}$ being handled for each term $t_i$,
incrementally accumulate three scores $Global(d)$, $Local(d)$, and  $RankScore(d)$ for each document $d$ visited,
and maintain three separate score-sorted queues $Q_{Gl}$, $Q_{Lo}$, and $Q_{Rk}$. 

\comments{
\begin{figure}[htbp]
\begin{center}
  \includegraphics[width=\columnwidth]{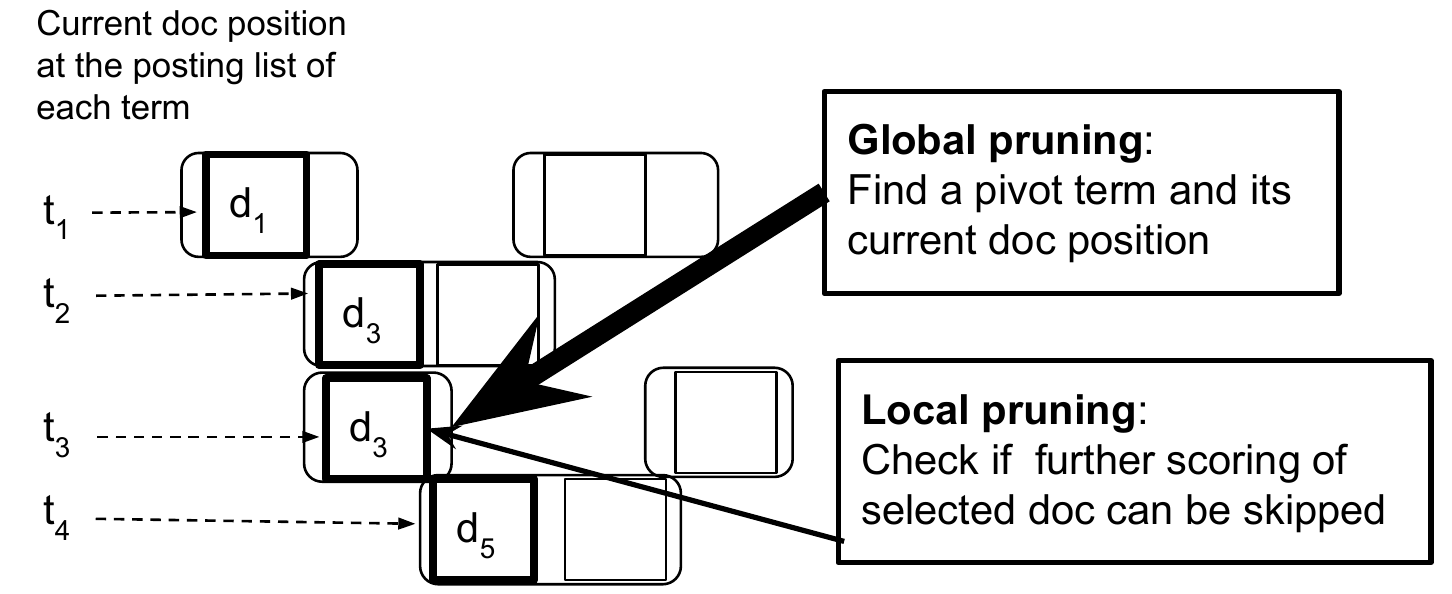}
\end{center}
  \caption{ Example of twol-level pruning in BMW}
  \label{fig:bmwex}
\end{figure}
}

\begin{figure}[htbp]
\begin{center}
  \includegraphics[width=0.5\columnwidth]{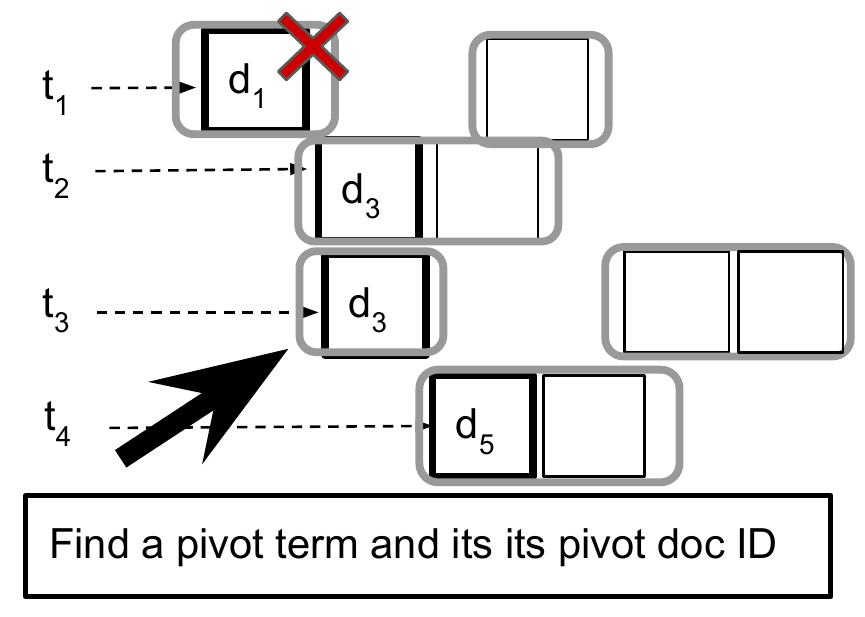}
\end{center}
  \caption{ Global pruning in BMW}
  \label{fig:BMWguidedskip}
\end{figure}

\begin{itemize}
[leftmargin=*]
\item {\bf Pruning at the global inter-document level with pivot identification.}
BMW~\cite{BMW} keeps a sorted search term list in each traversal step so that 
$d_{t_i} \leq d_{t_{i+1}}$ with $1\leq i \leq N-1$.
The pivot position that partitions these current document pointers is the smallest integer called $pivot$ such
that
$
\sum_{i=1}^{pivot}   \alpha \sigma_B[i] + (1-\alpha)  \sigma_L[i]  >  \theta_{Gl}.
$
This inequality means that  any document ID $d$ where $ d_{t_0} \leq d < d_{t_{pivot}}$  does not qualify for being
in the final top $k$ list based on  score $Global(d)$.
Then with the above pivot detection, for  $1 \leq i< pivot$, 
the current visitation pointer of the $i$-th posting list  moves to the closest block that
contains  a document ID equal to or  bigger than $d_{t_{pivot}}$. 

Figure~\ref{fig:BMWguidedskip} illustrates an example of global pruning in BMW with 4 terms and each posting list maintains a
pointer to the current  document being visited at a traversal step. Documents in each posting list are covered by a curved rectangle,
representing these lists are stored and compressed in a block-wise manner.
In the figure, the pivot identification at one traversal step locates  document $d_3$, and document IDs smaller than $d_3$ are skipped for any further consideration in
this traversal step.

\item {\bf Local pruning.} Let $d$ be the corresponding pivot document in pivot term $t_{pivot}$.
In Figure~\ref{fig:BMWguidedskip}, pivot term $t_{pivot}=t_3$  and $d=d_3$. 
A traversal procedure is executed to check if  detailed scoring of document $d$
can be avoided fully or partially and this procedure can be  similar as the one in the revised MaxScore algorithm described earlier.
As each posting list is packaged in a block manner in BMW, let $\Delta_B [x]$ and $\Delta_L[x]$ be the 
BM25 and learned block maximum weights of the block in the $x$-th posting list that contains $d$, respectively, and they are 0 if no such a block exists in this list.
The upper bound of $Local(d)$ can be tightened using the block-wise maximum weight instead of the list-wise maximum weight contributed by
each term as:
$
\sum_{i=1}^{N}   \beta \Delta_B [i]  +  (1-\beta)  \Delta_L[i].
$
When  decompressing the needed block of a posting list, 
the block-max contribution from the corresponding term  in the above expression can be replaced  by the actual BM25 and learned weights for 
document $d$. Then the upper bound of $Local(d)$ is further tightened, which can be directly compared with $\theta_{Lo}$ after every downward adjustment.
\end{itemize}

\begin{table}

\footnotesize
    \centering
\caption{ Guided  VBMW and MaxScore with uniCOIL on MS MARCO passages}


\label{tab:vbmw-relevance}

\begin{tabular}{c|l|lll}
\hline 
\multicolumn{1}{c|}{\textbf{Dataset}} & \multicolumn{1}{c|}{\textbf{Method}} & \multicolumn{1}{c}{\textbf{$k=10$}}   & \multicolumn{1}{c}{\textbf{$k=20$}}   & \multicolumn{1}{c}{\textbf{$k=100$}}  \\ \hline
\multirow{3}{*}{Dev} & MaxScore-2GTI     & 0.355$^\dag$, 2.6 (14.3) & 0.355$^\dag$, 3.4 (18.4) & 0.355, 5.5 (26.0) \\
& VBMW-2GTI     & 0.353$^\dag$, 4.3 (30.6) & 0.354$^\dag$, 5.2 (35.6) & 0.355, 8.6 (51.6) \\
& VBMW-GTI   & 0.339, 2.4 (14.2)  & 0.347, 3.0 (17.2) & 0.353, 5.4 (27.1) \\ \hline
\multirow{3}{*}{DL'19} & MaxScore-2GTI     & 0.714, 1.9 (12.9) & 0.713, 2.3 (14.2) & 0.713, 4.3 (18.6) \\
& VBMW-2GTI     & 0.708, 2.0 (20.0) & 0.708, 3.7 (23.2) & 0.710, 6.6 (33.1) \\
& VBMW-GTI      & 0.694, 1.7 (9.0)  & 0.700, 2.2 (11.7) & 0.710, 4.3 (17.9) \\ \hline
\multirow{3}{*}{DL'20} & MaxScore-2GTI     & 0.689, 2.8 (12.1) & 0.689, 3.3 (13.0) & 0.689, 5.3 (22.2) \\
& VBMW-2GTI     & 0.683, 3.9 (18.4) & 0.686, 4.9 (22.6) & 0.686, 8.4 (46.8) \\
& VBMW-GTI      & 0.676, 2.3 (10.5)  & 0.680, 2.9 (13.7) & 0.685, 5.3 (24.8) \\ \hline 
\end{tabular}
\end{table}

{\bf Evaluations on effectiveness  of 2GTI on VBMW}.
We choose uniCOIL to study the usefulness of VBMW-2GTI in searching the MS MARCO Dev set. 
SPLADE++ is not chosen because the test queries
are  long on average and MaxScore is faster than VBMW for such queries. 
Table~\ref{tab:vbmw-relevance} reports the performance  for VBMW-2GTI, VBMW-GTI, 
and MaxScore-2GTI for passage retrieval with uniCOIL when varying $k$.  
Each entry has a report format of $x, y (z)$ where $x$ is MRR@10 for Dev or NDCG@10 for DL'19 and DL'20.
$y$ is the MRT in ms, and $z$ is the $P_{99}$ latency in ms.
2GTI uses the fast setting with $\alpha=1$, $\beta=0.3$. For both 2GTI and GTI,  $\gamma=0.1$.
The result shows 2GTI provides a positive boost in relevance for  VBMW compared to GTI when $k$ is 10 and 20.
For $k=100$, the relevance difference is negligible.
MaxScore-2GTI is still faster   than VBMW-2GTI on average for all tested queries while their relevance difference is small.
we examine below if VBMW-2GTI can be useful for a subset of queries.


Table~\ref{tab:vbmw-Dev} reports the relevance and time  of these three algorithms 
in the passage Dev set for queries subdivided  based on their lengths and  if a query contains a stop word  or not.
That is for uniCOIL with $k=10$,   $\alpha=1$, $\beta=0.3$, and $\gamma=0.1$. 
Each entry has the same report format as in Table \ref{tab:vbmw-relevance}.
The result shows that 
VBMW-2GTI is much faster than MaxScore-2GTI for short queries ($k \leq 5$) that do not contain stop words and 
VBMW-2GTI has an edge in relevance over  VBMW-GTI while being very close to  MaxScore-2GTI for this class of queries.
The above result suggests  that   a fusion method can do well by
switching  the algorithm choice based on query characteristics and VBMW-2GTI can be used  for a class of queries.

\begin{table}

\footnotesize

    \centering
\caption{ Performance under different query classes with  $k=10$,  uniCOIL, and   MS MARCO passage Dev set}


\label{tab:vbmw-Dev}

\resizebox{1.05\columnwidth}{!}{

\begin{tabular}{l|llll}
\hline 
\textbf{QLength}           & $\leq 3$            & 4-5             & 6-7               & $\geq8$    \\ \hline \hline
\textbf{\# Q w/ SW }         & 113                & 1720              & 2175               & 2030                 \\ \hline
MaxScore-2GTI        & 0.286, 1.6 (12.2) & 0.376, 1.7 (8.9) & 0.347, 2.4 (11.5) & 0.315, 4.3 (20.9)  \\
VBMW-2GTI          & 0.289, 2.1 (15.2)          & 0.373, 2.2 (12.5)         & 0.346, 3.4 (16.1)          & 0.313, 8.5 (50.8)           \\ 
VBMW-GTI           & 0.257, 1.1 (6.1)          & 0.346, 1.4 (6.8)         & 0.334, 2.6 (12.2)          & 0.307, 7.4 (43.9)           \\ \hline \hline
\textbf{\# Q w/o SW}  & 327                & 445               & 130                & 40                   \\ \hline
MaxScore-2GTI      & 0.397, 1.3 (9.0)         & 0.430, 1.9 (9.4)        & 0.439, 3.7 (12.1)          & 0.558, 5.4 (15.9) \\
VBMW-2GTI          & 0.399, 0.9 (4.0)  & 0.430, 1.6 (6.3)  & 0.438, 3.6 (12.7) & 0.565, 5.6 (23.3)          \\ 
VBMW-GTI          & 0.385, 0.7 (3.4)  & 0.420, 1.3 (5.1) & 0.433, 2.6 (9.9) & 0.560, 4.1 (12.9)        \\ \hline 
\end{tabular}

}


\end{table}

\balance
\bibliographystyle{ACM-Reference-Format}
\normalsize
\typeout{} 
\bibliography{reference,bib/jinjin_thesis,bib/2022refer,CQbib/ranking,CQbib/mise,bib/2022extra,bib/url}
\end{document}